\documentclass{amsart}
\usepackage[margin=1in]{geometry}
\usepackage{hyperref}
\usepackage{graphicx}
\usepackage{amsmath}
\usepackage{amssymb}
\usepackage{amsthm}
\usepackage{dsfont}
\usepackage{mathrsfs}
\usepackage{sansmath}
\usepackage{xcolor}
\usepackage{comment}
\usepackage{booktabs}

\usepackage{mathtools}
\usepackage{enumitem}
\usepackage[mathscr]{euscript}
\newcommand{\euscr}[1]{\EuScript{#1}}

\usepackage{appendix}
\usepackage{bm, bbm}
\usepackage{caption}
\graphicspath{{figures/}}

\numberwithin{equation}{section}
\newtheorem{theorem}{Theorem}[section]

\newtheorem{definition}[theorem]{Definition}

\newtheorem{remark}[theorem]{Remark}
\newtheorem{example}[theorem]{Example}

\DeclareMathOperator*{\argmin}{arg\,min}

\DeclareMathOperator{\Var}{Var}

\newcommand{\R}{\mathbb{R}}

\newcommand{\dd}{\mathrm{d}}

\newcommand{\E}{\mathbb{E}}

\newcommand{\bR}{\mathbb{R}}

\DeclarePairedDelimiterX{\divg}[2]{(}{)}{%
  #1\;\delimsize\|\;#2%
}

\begin{document}

\title{Active Portfolio Management in Concentrated Equity Markets}

\author{Brian Ceco}
\address{Department of Mathematics, University of Toronto, Toronto, Canada}
\email{brian.ceco@mail.utoronto.ca}

\author{Xiaofei Shi}
\address{Department of Statistical Sciences, University of Toronto, Toronto, Canada}
\email{xf.shi@utoronto.ca}

\author{Ting-Kam Leonard Wong}
\address{Department of Statistical Sciences, University of Toronto, Toronto, Canada}
\email{tkl.wong@utoronto.ca}

\thanks{X.~Shi acknowledges support from the NSERC Discovery Grant RGPIN-2024-04569. T.-K.~L.~Wong acknowledges support from the NSERC Discovery Grant RGPIN-2025-06021.}

\keywords{Stochastic portfolio theory, market diversity, dispersion, transaction costs, portfolio optimization}
\subjclass[2020]{91G10, 91G15, 93E20}
\date{}

\begin{abstract}
The equal-weighted portfolio is a passive, rule-based strategy that has historically been difficult to outperform, delivering higher returns than the capitalization-weighted ``market'' benchmark across many markets and periods. Stochastic portfolio theory (SPT) reveals that this relative performance is regime dependent, with the equal-weighted portfolio underperforming during periods of increasing market concentration and high correlations, particularly market bubbles. These observations have motivated us to formulate and solve a stochastic control problem in which an investor actively allocates between the equal-weighted and market portfolios. The investor bases their allocation decisions on forecasts made under a flexible stochastic diversity--dispersion (SDD) model. 
Using a quadratic surrogate for implementation frictions, we characterize the optimal allocation through a linear forward--backward SDE and obtain an explicit ``aiming in front of a moving target'' representation of the optimal trading rate, in the spirit of G\^arleanu and Pedersen. The penalty parameters are calibrated in sample to match the cumulative wealth effect of proportional transaction costs, while out-of-sample performance is evaluated with those costs deducted directly from portfolio wealth. 
Using historical S\&P 500 data, we show that a mean-reverting SDD specification reproduces several empirical features of market diversity and dispersion. In out-of-sample backtests from 1995 to 2024, the resulting strategies deliver higher cumulative net returns than both the equal-weighted and market portfolios, and higher information ratios relative to the market than the equal-weighted portfolio after 15-basis-point proportional transaction costs.
\end{abstract}

\maketitle

\section{Introduction} \label{sec:intro}

Since the development of the Capital Asset Pricing Model (CAPM) in the 1960s \cite{L65, M66, SHP64, T61}, a central question in active portfolio management has been whether and how investors can consistently outperform the market. Building on the mean-variance framework of Markowitz \cite{MAR52}, the CAPM identifies the market portfolio as efficient, contributing to the popularity of index funds \cite{M99}. A surprising empirical finding, however, is that the simple equal-weighted portfolio has outperformed the market portfolio across a wide range of markets and sample periods. The equal-weighted (or ``$1/n$'') portfolio invests an equal proportion of wealth in each stock, and in particular is not susceptible to the forecasting error associated with classical mean-variance portfolios. DeMiguel et~al.~\cite{DGU09} demonstrated in an extensive comparative study that the equal-weighted portfolio delivered superior risk-adjusted performance compared to a wide variety of portfolio construction methods, including the market portfolio and Markowitz's mean-variance portfolio. At the same time, the relative performance of the equal-weighted portfolio can deteriorate sharply when market capitalization becomes more concentrated, such as during the dot-com bubble, or the current artificial intelligence boom. See \cite{CSW25, RX20, TM21} for recent empirical studies and \cite{DLN26, TCS26} for some implications of market concentration on active portfolio management.

Stochastic Portfolio Theory (SPT) \cite{FER02, FK09}, reviewed in Section~\ref{sec:prelim}, provides a useful framework for understanding the performance of portfolios relative to the market. SPT is built upon a multivariate market model, and aims to study and exploit  \emph{macroscopic} properties of equity markets from a descriptive rather than normative perspective. Under SPT, the relative log return of the equal-weighted portfolio with respect to the market portfolio decomposes into two terms: the change in a \emph{market diversity} process and an accumulated \emph{dispersion}, or excess growth, term. Diversity measures how evenly capitalization is distributed across firms, while dispersion captures the benefit of volatility harvesting \cite{BNW15, PW13} under systematic rebalancing. This decomposition not only explains why the equal-weighted portfolio underperforms when market concentration increases, but also suggests parsimonious modelling of interpretable macroscopic quantities rather than the full high-dimensional dynamics of all individual stocks \cite{CW25}. More generally, SPT characterizes the relative performance of a wide class of portfolios known as functionally generated portfolios.

\medskip

In this paper, we develop a general and tractable stochastic control framework to address the following natural and important question: can we use information about market diversity and dispersion to exploit the benefits of diversification and volatility harvesting, while protecting the portfolio from increases in market concentration? We achieve this in two steps:
\begin{itemize}
\item First, we introduce a \emph{stochastic diversity--dispersion (SDD) model} for the joint dynamics of market diversity and dispersion (Section~\ref{sec:sdd_model}). 
It provides an effective dimension reduction, preserving the key quantities that determine relative performance without requiring a full specification of the joint dynamics of all stocks. 
A tractable special case is the \emph{mean-reverting model} where diversity is mean-reverting with stochastic volatility and the logarithm of dispersion follows an Ornstein--Uhlenbeck (OU) process. 

\item Second, we formulate a tractable dynamic allocation problem in which an investor adjusts an equal-weighted tilt relative to the market portfolio. The objective balances expected relative log return and active risk against quadratic penalties on the size and adjustment rate of the tilt. These penalties serve as a surrogate for proportional transaction costs, with their coefficients calibrated in sample to match the cumulative active wealth drag generated by proportional costs. All out-of-sample performance is evaluated by deducting proportional costs directly from portfolio wealth. In Theorem~\ref{thm: FBSDE}, we characterize the resulting strategy through a linear forward--backward SDE and derive an explicit ``aiming in front of a moving target'' representation of the optimal trading rate, in the spirit of G\^arleanu and Pedersen~\cite{GP13}. For the mean-reverting specification, Theorem~\ref{thm: optimal control logou} further provides an explicit one-dimensional integral representation for the conditional forecasts entering the optimal strategy. Although we focus on the equal-weighted portfolio, the framework extends naturally to other functionally generated portfolios.
\end{itemize}
In Section~\ref{sec:empirics}, we perform a series of experiments to validate our approach with simulated and US stock data.\footnote{Our experiments can be reproduced via the following GitHub repository: \url{https://github.com/brianceco/apm-in-cem}.} We calibrate the mean-reverting model to historic S\&P 500 market diversity and realized dispersion, and show that it is able to reproduce several stylized facts documented in \cite{CSW25}. We investigate how our portfolios perform under simulated diversity and dispersion data, before backtesting the portfolio with real data under proportional transaction costs, using the methodology in \cite{RX20}. Under both the mean-reverting model and a discrete-time data-driven \emph{trending model} introduced in Section~\ref{sec:trending}, our portfolio is able to outperform both the capitalization-weighted and equal-weighted portfolios in the out-of-sample period 1995--2024, which contains several episodes of increasing market concentration. Appendix~\ref{sec:fama_french_test} reports descriptive Fama--French factor regressions. The strategies have positive alpha point estimates, but their statistical significance is not robust across factor specifications.

To the best of our knowledge, this is the first study to combine (i) an explicit stochastic model for the joint dynamics of market diversity and dispersion, (ii) optimization of a \emph{time-varying} functionally generated portfolio under (quadratic) trading-cost penalties, and (iii) a stochastic optimal control formulation within SPT. The closest related approaches are \cite{AJ18} and \cite{TM21}. Al-Aradi and Jaimungal \cite{AJ18} formulate a stochastic control problem that balances benchmark outperformance against tracking and permits functionally generated portfolios as benchmarks, but do not jointly model diversity and dispersion or penalize changes in the equal-weighted tilt. Taljaard and Maré \cite{TM21} use linear-regression forecasts to switch dynamically between the equal-weighted and market portfolios. In our terminology, their strategy resembles a bang--bang control and may therefore generate abrupt re-allocations. In contrast, our framework produces a continuously varying equal-weighted tilt that explicitly balances expected relative performance against active risk and the costs of maintaining and adjusting the position.

More broadly, much of the literature on optimizing functionally generated portfolios focuses on selecting a \emph{fixed} portfolio. For example, \cite{W15} studies a shape-constrained optimization problem over functionally generated portfolios, while \cite{CW22} uses regularized empirical risk minimization to optimize over a class of rank-based portfolios. Using machine-learning methods, \cite{SK16} selects the parameters of a diversity-weighted portfolio. Under ergodicity assumptions and asymptotic relative-growth objectives, \cite{itkin2022robust, IL24, kardaras2012robust} characterize optimal portfolios that are themselves functionally generated.

Although stochastic control has been central to continuous-time portfolio choice since Merton \cite{MER69}, its explicit application within SPT remains comparatively limited. In addition to \cite{AJ18}, relevant examples include \cite{AJ21}, which optimizes risk-adjusted active return in the presence of latent factors, and \cite{I25b}, which studies an investment-consumption problem in a rank-based market. A related empirical contribution is \cite{AFF07}, which develops a statistical model for forecasting market diversity from macroeconomic variables and applies it to a diversity-based strategy for the S\&P~500. We refer to \cite[Section~6.3]{CSW25} for a broader discussion of the literature.

\subsection{Organization}
The rest of the paper is organized as follows. Section~\ref{sec:prelim} reviews stochastic portfolio theory and derives the relative value of mixtures between the equal-weighted and market portfolios. Section~\ref{sec:sdd_model} introduces the stochastic diversity--dispersion model and the mean-reverting specification. Section~\ref{sec:stochastic_control} formulates and solves the stochastic control problem. Section~\ref{sec:empirics} presents the calibration and backtesting results. Appendix~\ref{sec:proofs} contains proofs of theoretical results, and Appendix~\ref{sec:fama_french_test} provides details of Fama--French regressions.

\subsection{Notation}
We let $(\Omega, \euscr{F}, \mathbb{F}=\{\euscr{F}_t\} _{t \in [0,T]}, \mathbb{P})$ denote a filtered probability space satisfying the usual conditions. We assume that $\mathbb F$ is the usual augmentation of the filtration generated by a $d$-dimensional Brownian motion $B$, with $d$ sufficiently large to support all Brownian motions used below.
All Brownian motions are understood to be $\mathbb F$-Brownian motions.  We let $\mathbf{L}^2_{\mathbb{F}}(T;\mathbb{R}^{n})$ denote the Hilbert space of processes in $L^2(\mathcal{B}([0,T])\otimes \euscr{F}, dt \times d\mathbb{P})$ which are progressively measurable. Finally, we let $\E_t[ \cdot]$ denote the conditional expectation given $\euscr{F}_t$.

\section{Stochastic portfolio theory and a stochastic diversity--dispersion model} \label{sec:prelim}
In Section~\ref{sec:SPT}, we review the part of stochastic portfolio theory that motivates our dynamic portfolio choice problem. The key idea is that if we restrict attention to time-varying mixtures of the equal-weighted portfolio and the market portfolio, then it suffices to track the joint dynamics of market diversity and dispersion; empirical stylized facts for these quantities were recently investigated in \cite{CSW25}. This drastically reduces the dimensionality of the optimization and naturally leads to our \emph{stochastic diversity--dispersion} (SDD) model, introduced in Section~\ref{sec:sdd_model}. More generally, the equal-weighted portfolio may be replaced by any functionally generated portfolio, with the corresponding notions of diversity and dispersion; see Remark~\ref{rmk:fgp}.

\subsection{Market diversity and dispersion}  \label{sec:SPT} 
To concentrate on the core financial ideas, we build our theory on a {\it closed market} where the investment universe is fixed throughout the horizon. Also, since the focus is on equity portfolio management, we do not include a risk-free asset or bonds explicitly in the model. For further details, we refer the reader to \cite{FER02}, which is the standard reference for this classic framework; also see \cite[Appendix 1]{CSW25} for a brief treatment. In Remark~\ref{rmk:leakage}, we include a technical discussion about the \emph{leakage effect} which arises when portfolios can only invest in the largest $n$ stocks of the market. This is relevant in our backtest since our benchmark S\&P 500 is approximately a rank-based capitalization-weighted portfolio ($n = 500$) not equal to the entire market. The assumption of market closedness can be partially relaxed in models of open markets \cite{IL24-2, KK21} and markets with a stochastic number of stocks \cite{BKT24, KS16}.

Let $n \geq 2$ be the total number of stocks, and let $T > 0$ be a finite time horizon.  Throughout this paper, the unit of time is one year unless otherwise stated. On a given filtered probability space satisfying the usual conditions, we let $S(t) = (S_1(t), \ldots, S_n(t))$ be an $n$-dimensional continuous semimartingale, with values in the positive quadrant $(0, \infty)^n$, that represents the market capitalizations of the stocks. More specifically, we assume that the vector $\log S(t) \coloneqq (\log S_1(t), \ldots, \log S_n(t))$ of log market capitalizations satisfies the stochastic differential equation (SDE)
\begin{equation}\label{eq:logprice}
d \log S(t) = \gamma(t)dt + \xi(t) dW(t), \quad S(0) = s_0 \in (0, \infty)^n,
\end{equation}
where $W(t)$ is an $n$-dimensional standard Brownian motion. Here, $\gamma(t) = (\gamma_i(t))$, the vector of growth rates, and $\xi(t) = (\xi_{ij}(t))$, the diffusion matrix, are progressively measurable processes satisfying suitable integrability conditions, see \cite[Chapter 1]{FER02}. Define $\sigma(t) \coloneqq \xi(t) \xi(t)^{\top}$, and note that $\sigma_{ij}(t) \dd t= \dd \langle \log S_i, \log S_j \rangle(t)$.

Let $\bar{S}(t) \coloneqq S_1(t) + \cdots + S_n(t)$ denote the total market capitalization. The probability vector 
\begin{equation} \label{eq:market.weight}
\mu(t)\coloneqq \frac{S(t)}{\bar{S}(t)}
\end{equation}
takes values in the open unit simplex $\Delta_n \coloneqq \{ p = (p_1, \ldots, p_n) \in (0, 1)^n, \sum_{i=1}^n p_i=1 \}$ and represents the vector of {\it market weights} of the stocks. It is helpful to think of $\mu(t)$ as the weights of the market portfolio (Example~\ref{eg:market.portfolio}). In practice, this may be proxied by a capitalization-weighted index such as the S\&P 500.

\emph{Market diversity} measures the concentration of market capitalization, and is often quantified by Shannon entropy \cite{F99}. In this paper, we use the diversity measure that generates the equal-weighted portfolio through \eqref{eq:fgp}, see Example~\ref{eg:equal.weighted.portfolio}.

\begin{definition}[Market diversity] \label{def:market.diversity}
We define \emph{market diversity} by
\begin{equation} \label{eq:market.diversity}
\varphi_{\mathrm{ew}}(t) \coloneqq \sum_{i = 1}^n \frac{1}{n} \log \mu_i(t).
\end{equation}
\end{definition}

Note that
\begin{equation} \label{eq:cross.entropy}
\varphi_{\mathrm{ew}}(t) = - H(\bm{1} / n) - H\divg{ \bm{1}/n }{\mu(t)},
\end{equation}
where $H$ and $H\divg{\cdot}{\cdot}$ denote Shannon entropy and relative entropy, respectively. Thus,  $\varphi_{\mathrm{ew}}(t)$ is the negative cross-entropy of $\bm{1}/n$ relative to $\mu(t)$, and can be equivalently written as 
$\varphi_{\mathrm{ew}}(t)= \log (\mu_1(t) \cdots \mu_n(t))^{1/n}$, which is the logarithm of the geometric mean of the market weights.
In Figure~\ref{fig:diversity} we plot the monthly time series of $\varphi_{\mathrm{ew}}$ for the S\&P 500 universe. We see that in recent years, market diversity has decreased to a level not previously seen since the dot-com bubble of the late 1990s. As was the case back then, this rise in market concentration is largely attributable to the overperformance of the technology sector, this time driven by advances in artificial intelligence and related technologies. At this point, it is unclear how long it will take, if ever, for diversity to ``revert'' to its historical average (about $-6.95$).

\begin{figure}[!htb]
       \centering
       \includegraphics{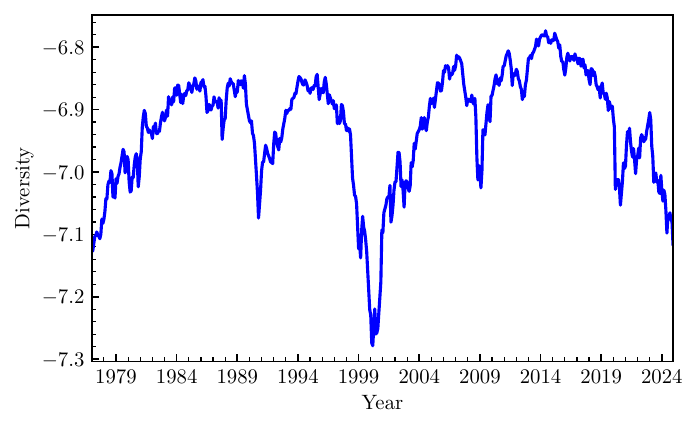}
       \vspace{-0.2cm}
       \caption{Monthly time series of market diversity for the S\&P 500 universe.
       }
       \label{fig:diversity}
\end{figure}

A {\it portfolio} is a progressively measurable process $\pi(t)$ satisfying $\bm{1}^{\top} \pi(t) \equiv 1$ and such that the SDE \eqref{wealthsde} below has a unique strong solution; it represents the vector of portfolio weights (rather than number of shares). The portfolio is said to be {\it long-only} if $\pi_i(t) \geq 0$ for all $1 \leq i \leq n$ and $0 \leq t \leq T$. The following two examples, corresponding to the two distributions in \eqref{eq:cross.entropy}, play crucial roles in our paper.

\begin{example} [Market portfolio] \label{eg:market.portfolio} 
The market portfolio is denoted by $\pi(t) \equiv \mu(t)$. %
\end{example}

\begin{example} [Equal-weighted portfolio] \label{eg:equal.weighted.portfolio}
The equal-weighted portfolio $\pi^{\mathrm{ew}}$ is given by $\pi^{\mathrm{ew}}(t) \coloneqq \frac{1}{n} \bm{1}$, where $\bm{1} \coloneqq (1, \ldots, 1)^{\top}$ is the vector of ones.
\end{example}

Given a portfolio $\pi$, its frictionless (gross or nominal) {\it wealth process} $Z^{\pi}$ is given as the solution of the SDE
\begin{equation}\label{wealthsde}
\frac{dZ^{\pi}(t)}{Z^{\pi}(t)} = \sum_{i = 1}^n \pi_i(t) \frac{d S_i(t)}{S_i(t)}, \quad Z^{\pi}(0) = z_0 \in (0, \infty)
,
\end{equation} 
where $z_0$ is the initial wealth. For example, the value of the market portfolio is given by
\[
Z^{\mu}(t) = z_0 \frac{\bar{S}(t)}{\bar{S}(0)}, \quad t \geq 0.
\]

Given the importance of the market portfolio as a benchmark for portfolio management, we use its value as the num\'{e}raire when evaluating the performance of other portfolios.

\begin{definition}[Relative value]
The relative value of a portfolio $\pi$ with respect to the market portfolio $\mu$ is defined by
\begin{equation} \label{eq:relative.value}
V^{\pi}(t) \coloneqq \frac{Z^{\pi}(t)}{Z^{\mu}(t)}, \quad \text{where } Z^{\pi}(0) = Z^{\mu}(0) \text{ so that } V^{\pi}(0) = 1.
\end{equation}
\end{definition}

\begin{remark}
In \eqref{wealthsde}, $Z^{\pi}$ denotes the nominal (dollar) value which neglects inflation. Its effect is automatically cancelled out in the ratio \eqref{eq:relative.value} which defines the relative value.
\end{remark}

In the continuous-time setting, there is a discrepancy between the (instantaneous) growth rate of a portfolio, i.e., the drift coefficient of $\log Z^\pi$, and the weighted average of the growth rates of the individual stocks due to the It\^o correction term. This motivates the notion of the \emph{dispersion} of a portfolio.

\begin{definition}[Dispersion]  \label{def:dispersion}
Given a portfolio $\pi$, we define its dispersion by
\begin{equation} \label{eq:dispersion}
\delta_{\pi}(t) \coloneqq \sum_{i = 1}^n \pi_i(t) \sigma_{ii}(t) - \sum_{i, j = 1}^n \pi_i(t) \pi_j(t) \sigma_{ij}(t).
\end{equation}
We denote by $\delta_{\mathrm{ew}}$ the dispersion of the equal-weighted portfolio.
\end{definition}

Our terminology is explained in Remark~\ref{rmk:dispersion}. The dispersion or equivalently the excess growth rate arises naturally in the dynamics of the portfolio value. Applying It\^{o}'s lemma to \eqref{wealthsde}, we have
\begin{equation*}
d \log Z^{\pi}(t) = \sum_{i = 1}^n \pi_i(t) d \log S_i(t) + \frac{1}{2} \delta_{\pi}(t) d t.
\end{equation*}
In fact, by a {\it num\'{e}raire invariance property} \cite[Lemma 1.3.4]{FER02}, we also have
\begin{equation} \label{eq:decomp}
d \log V^{\pi}(t) = \sum_{i = 1}^n \pi_i(t) d \log \mu_i(t) + \frac{1}{2} \delta_{\pi}(t) d t.
\end{equation}

For a long-only portfolio $\pi$, the dispersion $\delta_{\pi}(t)$ is non-negative, and can be interpreted as the difference between the weighted average instantaneous variance of the stocks and the instantaneous variance of the portfolio. As such, it may be regarded as a measure of the diversification of a portfolio as advocated in \cite{BF92}. Writing $\sigma_{ij}(t) = \sigma_i(t) \sigma_j(t) \rho_{ij}(t)$, where $\sigma_i(t) \coloneqq \sqrt{\sigma_{ii}(t)}$ and $\rho_{ij}(t)$ is the instantaneous correlation, shows that
\[
\delta_{\pi}(t) = \sum_{i = 1}^n \pi_i(t) (1 - \pi_i(t)) \sigma_{i}^2(t) - \sum_{i \neq j} \pi_i(t) \pi_j(t) \sigma_i(t) \sigma_j(t) \rho_{ij}(t)
\]
increases if correlations decrease. Thus, the dispersion captures how much the stocks move relative to each other. This is closely related to the work \cite{SR00} whose authors define dispersion as cross-sectional correlation.

\begin{remark}[Dispersion indices]\label{rmk:dispersion}
In September 2023 and May 2025 respectively, S\&P Dow Jones Indices and Cboe launched the implied and realized \emph{dispersion indices} for the S\&P 500 universe \cite{DSPX}. They provide backward and forward looking measures of cross-sectional volatility to aid investment decisions, and their definitions are closely related to the dispersion (or excess growth) of portfolios as defined above.

Given a portfolio $\pi$ and a window $h > 0$, consider, at each time $t$, the annualized variance
\begin{equation} \label{eq:cross.sectional.variance}
\frac{1}{h} \sum_{i = 1}^n \pi_i(t - h) (R_i(t-h, t) - R_{\pi}(t - h, t))^2,
\end{equation}
where $R_i(t - h, t)$ is the return of stock $i$ over $[t - h, t]$, and $R_{\pi}(t - h, t)$ is the return of the portfolio $\pi$. We let $r_i = \log(1 + R_i)$ and $r_{\pi} = \log(1 + R_{\pi})$ be the log returns. The realized dispersion index returns the (square root of) \eqref{eq:cross.sectional.variance} for the S\&P 500 universe and its weights, and $h$ equals $1$ or $30$ days. On the other hand, the implied dispersion index uses options and VIX data to compute a forecast of \eqref{eq:cross.sectional.variance} over the horizon $[t, t + h]$; theoretically, it corresponds to the expected value under a risk-neutral pricing measure.

For our purposes, we define the (annualized)  \emph{realized dispersion} of $\pi$ over a horizon $[t - h, t]$ by\begin{equation}\label{eq:realized_dispersion}
\mathrm{RD}_{\pi}(t, h) \coloneqq  \frac{1}{h}\sum\limits_{k = 1}^N \left( \sum\limits_{i=1}^{n} \pi_{i}(t_{k-1}) (r_{i}(t_{k-1}, t_k) - r_{\pi}(t_{k-1}, t_k))^2 \right),
\end{equation}
where $(t_k)_{k = 0}^N$ is a given partition of $[t - h, t]$. That is, we compute for each subinterval the variance of the log returns weighted by the portfolio, then sum over time and normalize. In our implementation, we use the partition corresponding to daily returns. Realized dispersion \eqref{eq:realized_dispersion} may be regarded as an estimator of the average dispersion
\[
\frac{1}{h} \int_{t - h}^t \delta_{\pi}(s) d s,
\]
and the approximation becomes exact as $\max_k |t_{k+1} - t_k| \downarrow 0$. This is analogous to the usual realized volatility as an estimator of integrated variance.

In Figure~\ref{fig:dispersion}, we plot the realized dispersion following \eqref{eq:realized_dispersion} for the equal-weighted portfolio, where $h = 1/12$ equals a month. We see clear clustering of volatility. In the inset of Figure~\ref{fig:dispersion}, we plot the monthly realized dispersion for the market portfolio (in red), together with the squared S\&P implied dispersion index (DSPX) \cite{DSPX} (in gray), for the more recent period 2014--2024. Discrepancies between implied dispersion and realized dispersion arise partly as a hedge against the risk of upward jumps in correlations. This market phenomenon is known as the \emph{correlation risk premium}, and efforts to predict it form the basis of \emph{dispersion trading} \cite{DMV09}.

\end{remark}

\begin{figure}[!ht]
       \centering
       \includegraphics{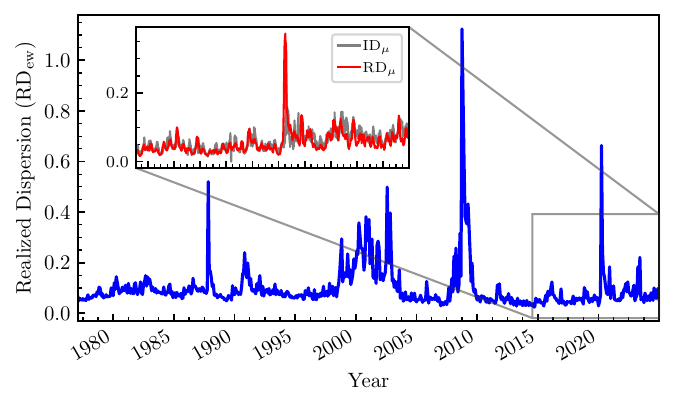}
       \caption{Monthly realized dispersion of the equal-weighted portfolio (blue) and the market portfolio (red) for the S\&P 500 universe. We contrast $\mathrm{RD}_\mu$ with implied dispersion $\mathrm{ID}_\mu$ (gray), with the latter given by the (squared) S\&P500 implied dispersion index (DSPX) \cite{DSPX}.}
       \label{fig:dispersion}
\end{figure}

Let $V^{\mathrm{ew}} \coloneqq V^{\pi^{\mathrm{ew}}}$ be the relative value of the equal-weighted portfolio. From \eqref{eq:decomp}, we see that
\begin{equation} \label{eq:ewactiveret}
\begin{split}
&d \log V^{\mathrm{ew}}(t) = d \varphi_{\mathrm{ew}}(t) + \frac{1}{2} \delta_{\mathrm{ew}}(t) dt, \qquad 
V^{\mathrm{ew}}(0) = 1. %
\end{split}
\end{equation}
Thus, the relative value $V^{\mathrm{ew}}$ is governed entirely by market diversity $\varphi_{\mathrm{ew}}$ and dispersion $\delta_{\mathrm{ew}}$. As discussed in Remark~\ref{rmk:leakage} below, \eqref{eq:ewactiveret} only holds approximately in practice because of constituent changes, dividends and discrete rebalancing. Nevertheless, the decomposition remains informative. Dispersion contributes a nonnegative drift, whereas changes in diversity can move relative value in either direction. Over short horizons, a sufficiently sharp decline in diversity can dominate the dispersion effect, causing the equal-weighted portfolio to underperform as market concentration increases. Motivated by this mechanism, \cite{TM21} uses a linear-regression signal to select either the equal-weighted portfolio or the market portfolio at each decision date. In this paper, we generalize the framework of \cite{TM21} to allow continuous adjustments between the equal-weighted and market portfolios. 

\begin{definition}[Equal-weighted tilt portfolio] \label{def:alpha.combination}
Given $\lambda$ in $\mathbf{L}_{\mathbb{F}}^{2}(T;\mathbb{R})$, we define the \emph{$\lambda$-tilt} towards the equal-weighted portfolio by
\begin{equation} \label{eq:alpha.combination}
\pi^{\lambda}(t) 
\coloneqq \mu(t) + \lambda(t) \left(\pi^{\mathrm{ew}}(t)-\mu(t)\right).
\end{equation}
We denote its relative value by $V^{\lambda} \coloneqq V^{\pi^{\lambda}}$. From \eqref{eq:ewactiveret} and It\^{o}'s formula, we have
\begin{equation} \label{eq:ewmasterformula}
d \log V^{\lambda}(t) = \lambda(t) d\varphi_{\mathrm{ew}} (t) + \frac{1}{2}\lambda(t)\delta_{\mathrm{ew}}(t)dt + \frac{1}{2}\lambda(t)(1-\lambda(t))d\langle \varphi _{\mathrm{ew}} \rangle(t), \qquad V^{\lambda}(0) = 1.
\end{equation}
\end{definition}
The representation~\eqref{eq:alpha.combination} can be interpreted as follows. Rather than representing merely a fraction of wealth allocated to the equal-weighted portfolio, $\lambda(t)$ determines the investor's exposure to the zero-net-investment long-short portfolio $\pi^{\mathrm{ew}}(t)-\mu(t)$. This long--short exposure reduces the portfolio's holdings of stocks with market weights above $1/n$ and increases its holdings of stocks with market weights below $1/n$. Thus, $\lambda(t)=0$ corresponds to the market portfolio and $\lambda(t)=1$ to the equal-weighted portfolio, while negative values reverse the long--short exposure and values greater than one amplify it. The return on $\pi^{\mathrm{ew}}(t)-\mu(t)$ is termed the \emph{equal-minus-value} return in \cite{HankeEtAl2019EqualWeightTilt}, where it is introduced as a factor for characterizing the diversification exposures of managed portfolios. Hanke and Quigley \cite{HankeQuigley2014Diversified} similarly combine the market portfolio with equal-stock and equal-sector portfolios, but use fixed portfolio-design coefficients. Related ``mirror-image portfolios'' relative to the market are studied in \cite{FKK05}; see also \cite[Section~8]{FK09}. In our framework, by contrast, the equal-weighted tilt $\lambda(t)$ is an endogenous, potentially state-dependent control, so both the magnitude and direction of the active exposure adjust as market conditions evolve.

\begin{remark}[Functionally generated portfolio and the master formula] \label{rmk:fgp}
The decomposition \eqref{eq:ewactiveret} can be generalized to a much wider class of portfolios. To any function $\varphi: \Delta_n \rightarrow \bR$ which is $C^2$ on an open neighborhood of $\Delta_n$ in $\bR^n$, we may define a portfolio $\pi$ by
\begin{equation} \label{eq:fgp}
\pi_i(t) \coloneqq \mu_i(t) \left(1 + \frac{\partial \varphi}{\partial p_i}(\mu(t)) - \sum_{j = 1}^n \mu_j(t) \frac{\partial \varphi}{\partial p_j}(\mu(t))\right), \quad i = 1, \ldots, n.
\end{equation}
We call $\pi(t)$ the \emph{portfolio generated  by} $\varphi$. By \cite[Theorem 3.1.5]{FER02}, its log relative value admits the master formula
\begin{equation} \label{eq:master.formula}
d \log V^{\pi}(t) = d \varphi(t) + d\Gamma_{\pi}(t), \quad d \Gamma_{\pi}(t) \coloneqq \frac{1}{2} \sum_{i, j = 1}^n \left( -\frac{\partial^2 \varphi}{\partial p_i \partial p_j} - \frac{\partial \varphi}{\partial p_i} \frac{\partial \varphi}{\partial p_j}  \right)(\mu(t)) d \langle \mu_i,\mu_j \rangle(t),
\end{equation}
where $\varphi(t)$ is a shorthand for $\varphi(\mu(t))$. The equal-weighted portfolio is generated in this sense by $\varphi_{\mathrm{ew}}$. In \eqref{eq:master.formula}, the drift process $\Gamma_{\pi}$ is non-decreasing whenever $\varphi$ is exponentially concave, i.e., $e^{\varphi}$ is concave \cite{PW16}. Most results of this paper can be generalized if we replace the equal-weighted portfolio by a functionally generated portfolio. %
\end{remark}

\begin{remark}[Extension to several functionally generated portfolios] \label{rmk:multiple}
Our framework can be further generalized to manage an affine combination of $K$ given functionally generated portfolios and the market portfolio. Suppose we are given $K$ functionally generated portfolios $\pi^{(k)}$ with generating functions $\varphi_k$ and drift processes $\Gamma_k$, $k = 1, \ldots, K$. Given a $K$-dimensional process $\lambda \in \mathbf{L}_{\mathbb{F}}^{2}(T; \R^K)$, with a light abuse of notation, the relative return of the $\lambda$-tilt functionally generated portfolio, i.e.,
\begin{align*}
\pi^{\lambda}(t) = \mu(t) + \sum\limits_{k=1}^{K} \lambda_k(t)\left(\pi^{(k)}(t)-\mu(t)\right), 
\end{align*}
is then given by
\begin{align*}
d \log V^{\lambda}(t)&= \sum\limits_{k=1}^{K} \lambda_k(t)\left(d\varphi_k(t)+ d\Gamma_k(t)\right) + \frac{1}{2}\sum_{k = 1}^K \lambda_k(t) (1 - \lambda_{k}(t))  d \langle \varphi_k\rangle(t) \\
&\quad - \sum_{1 \leq k < \ell \leq K} \lambda_k(t)\lambda_\ell(t)d\langle\varphi_k,\varphi_\ell\rangle(t). 
\end{align*}
Given $\lambda$,  this log relative value of the $\lambda$-tilt functionally generated portfolio only depends on the $2K$-dimensional process $(\varphi_1, \ldots, \varphi_K, \Gamma_1, \ldots, \Gamma_K)$. We may think of this set-up as an active portfolio management problem in which we have carried out a dimensionality reduction to single out the market portfolio and $K$ tradable zero-net-investment long-short portfolios, where $\varphi_k$ and $\Gamma_k
$ are trading signals that we aim to capture. Such a dimensionality reduction is common in portfolio management, for example via fundamental or statistical factor models \cite{FF93, ROSS76}. One may possibly consider signals beyond functionals of the market weights. For example, in \cite{KIM23} the master formula was extended to incorporate auxiliary economic variables such as market-to-book ratios.
\end{remark}

\begin{remark}[Leakage] \label{rmk:leakage}
In practice, many market indices only track the top segment of the market. For example, the S\&P 500 corresponds roughly to the largest 500 stocks of the US market. Using the framework of \cite[Chapter 4]{FER02}, suppose now that $\mu$ is the capitalization-weighted portfolio of the largest $n$ stocks in an ambient market with $N > n$ stocks, and $\pi^{\mathrm{ew}}$ is the equal-weighted portfolio of the largest $n$ stocks. By construction, both portfolios have to rebalance whenever stocks enter or leave the restricted universe. When this leakage effect is included, the decomposition \eqref{eq:ewactiveret} needs to be modified. Let $\tilde{\mu} \in \Delta_N$ be the market weight vector of the ambient market. By \cite[Theorem 4.2.1]{FER02} (also see \cite[Example 4.3.2]{FER02}), under a mild nondegeneracy condition,  the relative value $V^{\mathrm{ew}}$ of $\pi^{\mathrm{ew}}$ with respect to $\mu$ satisfies
\begin{equation} \label{eq:leakage}
d \log V^{\mathrm{ew}}(t) = d \varphi_{\mathrm{ew}}(t) + \frac{1}{2}\delta_{\mathrm{ew}}(t) dt - dL(t),
\end{equation}
where $\varphi_{\mathrm{ew}}(t)$ is the diversity of the top $n$ stocks, 
\begin{align*}
d L(t) \coloneqq \frac{1}{2} \left( \frac{1}{n} - \mu_{(n)}(t) \right) d \Lambda_{\log \tilde{\mu}_{(n)} - \log \tilde{\mu}_{(n+1)}}(t),
\end{align*}
$X_{(1)} \geq X_{(2)} \geq \cdots \geq X_{(N)}$ denotes the (reverse) order statistics of a vector $X$, and $\Lambda_Y$ is the semimartingale local time of the process $Y$ at $0$. While \eqref{eq:leakage} holds theoretically when the rank constraint is constantly enforced through frictionless trading, actual indices only update their constituents periodically. For theoretical tractability, we develop the stochastic control problem using \eqref{eq:ewmasterformula} rather than the modified master formula \eqref{eq:leakage}. See Section~\ref{sec:data} for a further discussion.%

\end{remark}

\subsection{Stochastic diversity--dispersion model}\label{sec:sdd_model}
Motivated by the master formula \eqref{eq:ewmasterformula}, we propose a \emph{stochastic diversity--dispersion (SDD) model} for the joint dynamics of market diversity and dispersion. It can be viewed as a dimensionality reduction of the original market model in Section~\ref{sec:SPT}, and circumvents the need to specify the joint dynamics of all stocks in the market.  While tailored to the equal-weighted portfolio, it can be extended to other functionally generated portfolios in view of Remark~\ref{rmk:fgp}. Henceforth, we denote $\varphi = \varphi_{\mathrm{ew}}$ and $\delta = \delta_{\mathrm{ew}}$. 
An SDD  model consists of the following components:
\begin{itemize}
\item A strictly positive continuous semimartingale $\delta(t)$ in $\mathbf{L}^2_{\mathbb{F}}(T;\mathbb{R})$ that represents market dispersion. 
\item An It\^{o} process $\varphi(t)$ that represents market diversity and is given by the strong solution to the following stochastic differential equation (SDE)
\begin{equation}\label{eq:SDD}
d \varphi (t)=  b_{\varphi }(t) dt + \sigma_\varphi (t)dW_{\varphi }(t),\quad \varphi (0)=\varphi_0,
\end{equation}
where $b_{\varphi}$ and $\sigma_{\varphi}$ are progressively measurable processes with $\sigma_{\varphi} > 0$, $W_{\varphi}$ is a Brownian motion, and $\varphi_0 \in \R$ is the initial value.
\end{itemize}
We assume the following integrability condition holds:
\begin{align}\label{cond:integrability}
\E\left[\int_0^T \delta^2(t) dt \right]+
\E\left[\int_0^T\sigma_\varphi^4(t) dt \right] + 
\E\left[\int_0^T\left(\frac{2b_\varphi(t) + \delta(t)}{\sigma_\varphi^2(t)}\right)^4 dt \right] <\infty. 
\end{align}
Additional conditions will be stated in Section~\ref{sec:stochastic_control} as needed.

This formulation deliberately leaves the joint dynamics of diversity and dispersion only partially specified. In particular, the diversity drift $b_{\varphi}$ and volatility $\sigma_\varphi$, and the dispersion process $\delta$, may be general adapted processes. We retain this generality for two reasons. First, the stochastic control problem can be solved at this level of generality, as shown in Section~\ref{sec:stochastic_control}. Second, the flexible specification is consistent with the model-free spirit of stochastic portfolio theory. 
Empirical implementation, however, requires explicit models for $\sigma_\varphi$, $b_{\varphi}$ and $\delta$, typically calibrated to historical data.

From~\eqref{eq:ewmasterformula}, under the SDD model \eqref{eq:SDD}, the log relative value of the equal-weighted tilt portfolio satisfies 
\begin{equation} \label{eq:SDD.log.value}
d \log V^{\lambda}(t) = \lambda(t) \left( b_{\varphi}(t) + \frac{1}{2} \delta(t) + \frac{1}{2}(1 - \lambda(t)) \sigma_{\varphi}^2(t)\right) dt + \lambda(t) \sigma_{\varphi} (t) d W_{\varphi}(t).
\end{equation}

As a first approximation, we may regard both diversity and (in log scale) dispersion as mean reverting processes. This leads to the following instance of the SDD model which serves as our main example throughout the paper.

\begin{example}[Mean-reverting model]\label{eg:logou}
The \emph{mean-reverting model} is defined by the SDE \begin{equation}
\begin{cases}\label{eq:logou}
&d \varphi (t)=  \kappa_{\varphi }(\overline{\varphi}-\varphi (t))dt +  \nu _{\varphi } \sqrt{\delta(t)}dW_{\varphi }(t),\qquad \varphi (0)=\varphi_0,\\
&d \log \delta (t) = \kappa_{\delta}(\log \overline{\delta}-\log\delta  (t))dt + \nu_{\delta} dW_{\delta}(t), \quad  \delta(0)=\delta_0.
\end{cases}
\end{equation} 
Here, $b_\varphi(t) = \kappa_\varphi(\bar\varphi - \varphi(t))$ and $\sigma_\varphi (t) = \nu_\varphi \sqrt{\delta(t)}$. $W_{\delta}$ is another Brownian motion with 
\[
d \langle W_{\varphi}, W_{\delta}\rangle(t) = \rho dt. 
\]
The parameters of the model are $\kappa_{\varphi}, \nu_{\varphi}, \kappa_{\delta}, \overline{\delta}, \nu_{\delta}, \delta_0 > 0$, $\overline{\varphi}, \varphi_0 \in \R$ and $\rho \in [-1, 1]$. Under this model, market diversity is mean-reverting with stochastic volatility, and the logarithm of market dispersion is an Ornstein--Uhlenbeck (OU) process. Such a stochastic volatility model was first considered in the context of option pricing by Hull and White \cite{HW87}. It is fairly straightforward to see that \eqref{eq:logou} admits a strong solution given in closed-form by
\begin{align}
\begin{cases}                
\varphi (t) = \varphi _{0} e^{-\kappa_{\varphi }t} + \overline{\varphi }(1-e^{-\kappa_{\varphi }t}) + \nu _{\varphi } \int_{0}^{t} e^{-\kappa_{\varphi }(t-s)}\sqrt{ \delta(s)}\,dW_{\varphi }(s), \\
\delta(t) = \exp \left(  \log \delta_{0} e^{-\kappa_{\delta}t}+  \log \overline{\delta}  (1-e^{-\kappa_{\delta}t})+  \nu_\delta  \int_{0}^{t} e^{-\kappa_\delta(t-s)}\,dW_{\delta}(s)\right).
\end{cases}
\end{align} 
In this model, market diversity is mean reverting, and market dispersion affects the volatility of market diversity via the coefficient $\sqrt{\delta(t)}$. Thus, when dispersion is large market diversity becomes more volatile. This is a simple yet effective way to capture empirical relations between the two quantities \cite[Section 4]{CSW25}. The empirical fit of this model will be examined in Section~\ref{sec:log.OU.calibration}.
\end{example}

The mean-reverting model is Markovian. By construction, it does not capture regime switching which might be happening due to the rise of artificial intelligence. In Section~\ref{sec:trending}, we also consider a \emph{trending model} which is specified in discrete time and is intended to capture short-term dynamics. By using more sophisticated instances of the SDD model that incorporate data beyond capitalization weights and returns, as well as regime switching, we believe better performance can be achieved.

\section{Optimizing outperformance: a dynamic portfolio choice problem in SPT}\label{sec:stochastic_control}

In this section, based on the general SDD model~\eqref{eq:SDD}, we formulate and solve the dynamic allocation problem for the equal-weighted tilt $\lambda$. 
\subsection{Frictionless optimization}\label{ssec: frictionless market}
As a benchmark, we first consider the frictionless version of the portfolio optimization problem, where the agent chooses her optimal equal-weighted tilt $\lambda(t)\in  \mathbf{L}_{\mathbb{F}}^{2}(T;\mathbb{R})$ to maximize the terminal relative value $V^\lambda$ of her portfolio compared to the market portfolio, penalizing the corresponding active risk. In the absence of trading costs, using the master formula \eqref{eq:SDD.log.value}, this amounts to maximizing the criterion
\begin{align}\label{eq: frictionless target}
{J}_0(\lambda) 
&\coloneqq \E \left[ \log V^{\lambda}(T) - \frac{\gamma}{2} \langle \log V^{\lambda} \rangle (T) \right] \\
& = 
\E\left[\int_0^T \lambda(t)\left( b_\varphi(t) + \frac{1}{2}\delta(t)+\frac{1}{2}\sigma_\varphi^2(t) - \frac{1}{2} (1+\gamma)\sigma_\varphi^2(t)\lambda(t)\right)\ dt \right],
\notag
\end{align}
for some choice of risk aversion parameter $\gamma>0$. Given that the volatility process $\sigma_\varphi$ is strictly positive for every $t\in[0,T]$, the optimal equal-weighted tilt is readily determined by pointwise optimization as
\begin{align}\label{eq: frictionless alpha}
\lambda_0^*(t) \coloneqq \frac{b_\varphi(t) + \frac{1}{2}\delta(t)+\frac{1}{2}\sigma_\varphi^2(t)}{(1+\gamma )\sigma_\varphi^2(t)}  
= \frac{1}{1+\gamma}\left(\frac{2b_\varphi(t) +\delta(t)}{2\sigma_\varphi^2(t)} +\frac{1}{2}\right). 
\end{align}
Given the integrability condition~\eqref{cond:integrability}, it is straightforward to verify that $\lambda_0^*$ is admissible.

The optimal equal-weighted tilt, $\lambda_0^*$, is governed by the instantaneous expected log return of the equal-weighted portfolio relative to the market, $b_\varphi(t)+\tfrac{1}{2} \delta(t)$, scaled by its active variance $\sigma_\varphi^2(t)$. The additional $1/2$ term in~\eqref{eq: frictionless alpha} arises from the quadratic-variation term generated by mixing the equal-weighted and market portfolios. Thus, the target $\lambda_0^*$ adjusts both the magnitude and direction of the equal-weight exposure as market conditions evolve; it may become negative when concentration pressure is sufficiently strong and may exceed one when the expected diversification benefit is sufficiently large.
In particular,
\begin{align*}
\lambda_0^*(t)>0
\quad\Longleftrightarrow\quad
\frac{2b_\varphi(t) +\delta(t)}{\sigma_\varphi^2(t)} >-1.
\end{align*}
When this condition holds, the strategy takes a positive position in the long–short portfolio to capture the diversity premium. 
Moreover,  
\begin{align*}
\lambda_0^*(t)>1
\quad\Longleftrightarrow\quad
\frac{2b_\varphi(t) +\delta(t)}{\sigma_\varphi^2(t)}>2\gamma+1, 
\end{align*}
which means the strategy leverages the long–short portfolio to exploit the diversification benefit more aggressively. Conversely, 
\begin{align*}
\lambda_0^*(t)<0
\quad\Longleftrightarrow\quad
\frac{2b_\varphi(t) +\delta(t)}{\sigma_\varphi^2(t)}<-1,
\end{align*}
in which case the concentration effect dominates, and the optimal tilt becomes negative. 
Thus, unlike a static equal-weighted portfolio, $\lambda_0^*$ adjusts in either direction as market diversity evolves.

Notice that criterion $J_0$~\eqref{eq: frictionless target} yields the same optimal equal-weighted tilt as would be obtained if we instead optimize the following:
\begin{align}\label{target: frictionless adapted}
\tilde{J}_0(\lambda) 
&=\E\left[\int_0^T  \lambda(t)\left(\frac{ 2b_\varphi(t) + \delta(t)}{2\sigma_\varphi^2 (t)}+\frac{1}{2} - \frac{1}{2} (1+\gamma) \lambda(t)\right)\ dt \right].
\end{align}
We include the variance normalization by $\sigma_\varphi^2$ in the integrand of the objective functional to ensure analytical tractability when incorporating trading costs; see Remark~\ref{rmk:objective}.

\subsection{Frictional optimization}\label{ssec: frictional market}
The frictionless optimizer \(\lambda_0^*\) in \eqref{eq: frictionless alpha} is myopic: at each time \(t\), it depends only on the contemporaneous drift and variance of diversity and on dispersion, which together determine the local trade-off between expected relative log return and active risk. It therefore provides a natural frictionless target. Once implementation costs are taken into account, however, tracking this time-varying target closely can require frequent rebalancing, generating substantial turnover and reducing net performance.

Our objective in this section is to obtain a tractable policy rather than a structural model of transaction costs. Directly incorporating proportional costs at the constituent level would produce a nonsmooth, path-dependent control problem and would generally destroy the linear FBSDE structure used below. We therefore introduce a quadratic surrogate for implementation frictions. The surrogate is not intended to approximate the proportional cost function trade by trade. Instead, its coefficients will be calibrated in Section~\ref{sssec:backtest design} so that its cumulative penalty matches the active wealth drag generated by proportional transaction costs over the in-sample period. The resulting policy is then evaluated out of sample under proportional costs themselves.

We denote the investor's equal-weighted tilt by the process $\lambda$. To regulate this trading, we restrict $\lambda$ to be absolutely continuous with respect to time, i.e.,
\begin{align*}
d\lambda(t)=\dot{\lambda}(t)\,dt,
\end{align*}
where $\dot{\lambda}$ is the trading rate. The quadratic surrogate consists of two components. The first penalizes a nonzero equal-weighted tilt, whose underlying constituents must be continually rebalanced; its coefficient is $\Lambda_1>0$. The second penalizes changes in the tilt and therefore proxies for the additional turnover generated by dynamic allocation; its coefficient is $\Lambda_2>0$. We define the cumulative surrogate penalty by
\begin{equation}\label{eq:tc_penalty}
    d\mathit{TC}^\lambda(t)
    \coloneqq
    \frac{1}{2}
    \left(
        \Lambda_1\lambda(t)^2+\Lambda_2\dot{\lambda}(t)^2
    \right)
    d\langle\log V^{\mathrm{ew}}\rangle(t).
\end{equation}
We use $TC^\lambda$ to denote the surrogate penalty, not the transaction costs deducted from portfolio wealth. The parameters $\Lambda_1$ and $\Lambda_2$ are reduced-form penalty parameters rather than structural estimates of bid--ask spreads or commissions.
Recall that
\[
d\langle \log V^{\mathrm{ew}}\rangle(t)
=
\sigma_\varphi^2(t)\,dt.
\]
Thus, the specification in \eqref{eq:tc_penalty} scales both penalties
by the instantaneous active variance of the equal-weighted portfolio.
Applying the same normalization as in the frictionless problem then
yields the constant coefficients \(\Lambda_1\) and \(\Lambda_2\).

Fix an initial tilt $\lambda^0\in\mathbb{R}$. We consider the equal-weighted tilt process following
\[
\lambda(t)=\lambda^0+\int_0^t\dot{\lambda}(s)\,ds,
\qquad
\dot{\lambda}\in\mathbf{L}_{\mathbb{F}}^2(T;\mathbb{R}),
\]
and satisfying the integrability condition
\begin{align}\label{frictional cond: integrability}
\E\left[\langle\log V^\lambda\rangle(T)\right]
+\frac{\Lambda_1}{2}\E\left[\int_0^T\lambda(t)^2 d\langle\log V^{\mathrm{ew}}\rangle(t)\right]
+\frac{\Lambda_2}{2}\E\left[\int_0^T\dot{\lambda}(t)^2d\langle\log V^{\mathrm{ew}}\rangle(t)\right]<\infty.
\end{align}
We denote the resulting admissible class by $\mathcal{A}$.
Normalizing the local reward and cost rates by the instantaneous quadratic
variation of $\log V^{\mathrm{ew}}$ gives the frictional counterpart of $\tilde{J}_0$
\eqref{target: frictionless adapted}:
\begin{align}\label{eq: target}
J(\lambda)
&\coloneqq
\tilde{J}_0(\lambda)-\frac{1}{2} \E\left[\int_0^T\left(\Lambda_1\lambda(t)^2+\Lambda_2\dot{\lambda}(t)^2\right)dt\right]\notag\\
&=\frac{1+\gamma}{2}\E\left[\int_0^T\lambda_0^*(t)^2\,dt\right]
-\frac{1}{2}\E\left[\int_0^T\left((1+\gamma)\bigl(\lambda(t)-\lambda_0^*(t)\bigr)^2 +\Lambda_1\lambda(t)^2+\Lambda_2\dot{\lambda}(t)^2 \right) dt\right].
\end{align}
The second equality follows by completing the square in the frictionless criterion. Since the first term is independent of \(\lambda\), it may be omitted without affecting the optimizer. The remaining objective reveals three competing forces: the tracking term pulls \(\lambda\) toward the frictionless target \(\lambda_0^*\), the position penalty \(\Lambda_1\lambda^2\) shrinks the tilt toward the market portfolio \(\lambda=0\), and the adjustment penalty \(\Lambda_2\dot{\lambda}^2\) discourages rapid trading.

Under this quadratic surrogate, the optimal equal-weighted tilt is no longer myopic. The G\^ateaux differential of $J$ vanishes at the optimal trading rate $\dot\lambda\in \mathbf{L}_{\mathbb{F}}^{2}(T;\mathbb{R})$, i.e., for every arbitrary admissible $\dot\beta\in \mathbf{L}_{\mathbb{F}}^{2}(T;\mathbb{R})$, we have
\begin{align*}
0 &= \lim_{\varepsilon\to0}\frac{1}{\varepsilon}\left(J(\lambda+\varepsilon\beta)-J(\lambda)\right)    
\\&=
\E\left[\int_0^T \left( \beta(t)\left({(1+\gamma)} \lambda_0^*(t) - (1+\gamma+\Lambda_1)\lambda(t)\right) - \Lambda_2 \dot\lambda(t)\dot\beta(t) \right) dt\right]
\\&=
\E\left[\int_0^T \dot\beta(t)\left(\int_t^T \left((1+\gamma) \lambda_0^*({u}) - (1+\gamma+\Lambda_1)  \lambda(u)\right) du - \Lambda_2   \dot\lambda(t)\right)  dt\right].
\end{align*}
Together with the frictionless optimal equal-weighted tilt $\lambda_0^*$ defined by~\eqref{eq: frictionless alpha}, the first-order condition is therefore
\begin{align}\label{foc: alpha}
\dot\lambda(t) 
&= \frac{1}{\Lambda_2 }\E_t\left[\int_t^T \left( (1+\gamma)\lambda_0^*({u}) - (1+\gamma+\Lambda_1) \lambda(u)\right) du\right] \notag\\
&= \frac{(1+\gamma+\Lambda_1)}{\Lambda_2 } \E_t\left[\int_t^T\left(  \frac{1+\gamma}{1+\gamma+\Lambda_1}\lambda_0^*(u) -  \lambda(u)\right)  du\right]. 
\end{align}
The optimal trading rate and the associated optimal equal-weighted tilt are jointly described via a forward-backward stochastic differential equation (FBSDE) system, whose existence and uniqueness are given by the following theorem, the proof of which is given in Appendix~\ref{app:thm proof}.

\begin{theorem}\label{thm: FBSDE}
Suppose that the integrability condition~\eqref{cond:integrability} holds and that
$\Lambda_1,\Lambda_2>0$. For any initial tilt $\lambda^0\in\mathbb{R}$, the forward-backward system
\begin{equation}\label{fbsde:constant}
\left\{
\begin{aligned}
\lambda(t)
&=\lambda^0+\int_0^t\dot{\lambda}(u)\,du,\\
\dot{\lambda}(t)
&=\frac{1+\gamma+\Lambda_1}{\Lambda_2}\E_t\left[\int_t^T\left(\frac{1+\gamma}{1+\gamma+\Lambda_1}\lambda_0^*(u)-\lambda(u)\right)du\right],
\end{aligned}
\right.
\end{equation}
admits a unique adapted solution $(\lambda^*,\dot{\lambda}^*)\in\mathbf{L}_{\mathbb{F}}^2(T;\mathbb{R})^2$. 
Define
\begin{align}\label{eq: G}
G(t)
\coloneqq
\cosh\left(\sqrt{\frac{1+\gamma+\Lambda_1}{\Lambda_2}}\,(T-t)\right),
\end{align}
and for $0\leq t<T$, the forward-looking aim process by
\begin{equation}\label{eq: trading aim}
\lambda^{\mathrm{aim}}(t)
\coloneqq
-\frac{G(t)}{G'(t)}\E_t\left[\int_t^T\frac{1+\gamma}{\Lambda_2}\frac{G(u)}{G(t)}\lambda_0^*(u)\,du\right].
\end{equation}
Then the optimal equal-weighted tilt admits the closed-form representation
\begin{equation}\label{eq: optimal control}
\lambda^*(t)
=\frac{G(t)}{G(0)}\lambda^0-\int_0^t\frac{G'(s)G(t)}{G^2(s)}\lambda^{\mathrm{aim}}(s)\,ds,
\end{equation}
and the optimal trading rate is
\begin{equation}\label{eq: optimal trading}
\dot{\lambda}^*(t)
=\sqrt{\frac{1+\gamma+\Lambda_1}{\Lambda_2}}
\tanh\left(\sqrt{\frac{1+\gamma+\Lambda_1}{\Lambda_2}}\,(T-t)\right)
\bigl(\lambda^{\mathrm{aim}}(t)-\lambda^*(t)\bigr), \qquad t < T.
\end{equation}
It can be verified that
$\lim_{t \uparrow T} \dot{\lambda}^*(t) = 0$.
\end{theorem}

\begin{remark}
In our simulation experiments in Section~\ref{sec:optimal_control_simulations}, we initialize the tilt process by $\lambda^*(0) = \lambda_0^*(0)$. Another reasonable choice is $\lambda^0 = 0$, so that $\pi^{\lambda}(0) = \mu(0)$. In principle, we may perform another optimization over $\lambda^0$ to initialize the tilt process.
\end{remark}

The investor, therefore, trades toward a weighted forecast of future frictionless tilts rather than responding only to the current value of $\lambda_0^*$. For fixed $t$, the deterministic factor $G(u)/G(t)$ decreases with $u$, assigning less weight to more distant forecasts. The market state variables $b_\varphi$, $\sigma_\varphi$, and $\delta$ affect the policy through $\lambda_0^*$ and its conditional forecasts, whereas $\gamma$, $\Lambda_1$, and $\Lambda_2$ determine the horizon-dependent adjustment speed in \eqref{eq: optimal trading}. This separation gives rise to the ``aiming in front of a moving target'' mechanism which is familiar from quadratic trading models \cite{BSV17, GP13, GP16}, but is new in the context of stochastic portfolio theory.

\begin{remark} \label{rmk:objective}
The equivalence between the regularized and unregularized frictionless problems noted after \eqref{eq: frictionless alpha} does not generally extend to the frictional setting. In the frictionless problem, optimization is pointwise, so multiplying the instantaneous objective by the positive factor $\sigma_\varphi^2(t)$ leaves its maximizer unchanged. Once changes in $\lambda$ are penalized, however, portfolio choices across time become linked. The stochastic factor $\sigma_\varphi^2(t)$ then changes the relative weighting of the objective across dates and market states and may therefore alter the optimal tilt process. Under suitable moment and nondegeneracy conditions on $\sigma_\varphi^2(t)$, one may instead consider the unnormalized criterion
\begin{align*}
\widehat{J}(\lambda) 
\coloneqq \E\Bigg[&\log V^\lambda(T)-\frac{\gamma}{2}\langle\log V^\lambda\rangle(T) \\
&-\frac{\Lambda_1}{2}\int_0^T\lambda(t)^2d\langle\log V^{\mathrm{ew}}\rangle(t) -\frac{\Lambda_2}{2}\int_0^T\dot{\lambda}(t)^2d\langle\log V^{\mathrm{ew}}\rangle(t)\Bigg].
\end{align*}
The corresponding first-order condition has stochastic coefficients inherited from $\sigma_\varphi$. Although existence and uniqueness can still be established under suitable assumptions, the deterministic kernel in \eqref{eq: optimal trading} is generally lost, and the optimizer need not admit a closed-form representation. We therefore retain \eqref{eq: target} as the
tractable benchmark.

Two modelling choices should therefore be distinguished. First, the quadratic penalty is a tractable surrogate for the wealth effect of proportional implementation costs. %
Second, the variance normalization leading to \eqref{eq: target} is used to obtain a deterministic adjustment kernel and an explicit policy. Neither step implies that the optimizer of \eqref{eq: target} coincides with the optimizer of a portfolio problem containing proportional costs directly. The empirical link is established instead through the in-sample effect calibration in Section~\ref{sssec:backtest design} and through out-of-sample evaluation under proportional costs.
\end{remark}

To implement the frictional optimal control given by Theorem~\ref{thm: FBSDE}, we need to have conditional forecasts $\E_t[\lambda_0^*(u)]$, or equivalently forecasts of $\E_t[(2b_\varphi(u)+\delta(u))/\sigma_\varphi^2(u)]$. For the mean-reverting specification introduced in Example~\ref{eg:logou}, these forecasts admit an explicit integral representation.

\begin{theorem}\label{thm: optimal control logou}
Suppose that $(\varphi,\delta)$ follows the mean-reverting model
\eqref{eq:logou}. For $0\leq t\leq u\leq T$, define
\begin{align}
    m_\delta(t,u)
    &\coloneqq
    e^{-\kappa_\delta(u-t)}\log\delta(t)
    +\bigl(1-e^{-\kappa_\delta(u-t)}\bigr)\log\overline{\delta},
    \label{eq: logou conditional mean}\\
    q_\delta(t,u)
    &\coloneqq
    \frac{\nu_\delta^2}{2\kappa_\delta}
    \bigl(1-e^{-2\kappa_\delta(u-t)}\bigr),
    \label{eq: logou conditional variance}
\end{align}
and, for $0\leq t\leq s\leq u$,
\begin{equation}\label{eq: logou conditional covariance}
    C_\delta(t;s,u)
    \coloneqq
    \frac{\nu_\delta^2}{2\kappa_\delta}
    \left(
        e^{-\kappa_\delta(u-s)}
        -e^{-\kappa_\delta(u+s-2t)}
    \right).
\end{equation}
We have
\begin{align}
&\E_t\left[\delta(u)^{-1}\right]
= \exp\left(-m_\delta(t,u)+\frac{1}{2}q_\delta(t,u)\right), \label{eq: logou inverse moment}\\
&\E_t\left[\sqrt{\delta(s)}\,\delta(u)^{-1}\right]
= \exp\Bigg( \frac{1}{2}m_\delta(t,s)-m_\delta(t,u)+\frac{1}{2}\left[ \frac{1}{4}q_\delta(t,s)
+q_\delta(t,u)-C_\delta(t;s,u)\right]\Bigg). \label{eq: logou mixed moment}
\end{align}
Then, the conditional expectation of $\lambda_0^*$ is given by
\begin{align}\label{eq: ce log ou}
\E_t\left[\lambda_0^*(u)\right]
&=\frac{\kappa_\varphi\bigl(\overline{\varphi}-\varphi(t)\bigr) e^{-\kappa_\varphi(u-t)}}{(1+\gamma)\nu_\varphi^2}\E_t\left[\delta(u)^{-1}\right]\notag\\
&\qquad +\frac{\kappa_\varphi\nu_\delta\rho}{(1+\gamma)\nu_\varphi}
\int_t^u e^{-(\kappa_\varphi+\kappa_\delta)(u-s)}\E_t\left[\sqrt{\delta(s)}\,\delta(u)^{-1}\right] \, ds 
+\frac{1+\nu_\varphi^2}{2(1+\gamma)\nu_\varphi^2}.
\end{align}
\end{theorem}
Equation \eqref{eq: ce log ou} follows from the joint conditional Gaussianity of
$\log\delta(s)$ and $\log\delta(u)$, together with an integration-by-parts
argument for the component of $W_\varphi$ correlated with $W_\delta$. Substitution
into \eqref{eq: optimal trading} gives an explicit representation of the optimal
trading rate under the mean-reverting model. The proof is in Appendix~\ref{app:log ou proof}.  
This separation is useful in practice: the
portfolio rule is determined by the frictional optimization problem, whereas the
forecasting model for diversity and dispersion can be chosen to match the investment horizon and market behaviours.

\section{Numerical and empirical analysis}
\label{sec:empirics}
In this section, we conduct a series of numerical and empirical experiments to illustrate the main features of the SDD model and the performance of the associated stochastic control framework. Section~\ref{sec:data} describes the data we use and the construction of the empirical inputs. In Section~\ref{sec:log.OU.calibration}, we calibrate the mean-reverting model to S\&P 500 data and assess the extent to which it reproduces several stylized properties of market diversity and dispersion. We then use Theorem~\ref{thm: optimal control logou} to study the behaviour of the optimal tilt process \(\lambda^*\) under the calibrated model. Finally, in Section~\ref{sec:backtest}, we carry out a real-data backtest of the optimal control \eqref{eq: optimal control} under the mean-reverting model as well as a data-driven trending SDD model described in Section~\ref{sec:trending}.

\subsection{Data and empirical setup} \label{sec:data}
For the empirical analysis in this section, we take the investment universe to be the time-varying constituents of the S\&P 500. The S\&P 500 is a capitalization-weighted index comprising roughly the largest 500 publicly traded companies within the US, representing about 80\% of the country's equity market capitalization. We obtain constituent-level market capitalization and return data for the period 1977--2024 from the CRSP database,\footnote{See \url{https://www.crsp.org/products/research-products/crsp-us-stock-databases}.} and process the data using the publicly available code in \cite[Notebook 6]{RUF23}.\footnote{In addition, we manually correct the erroneous delisting return of -0.99 on October 14, 1977 for PERMNO 16731 to $0.005$, an approximation reflecting the acquisition price documented in \cite{sec1977newsdigest}.}
We discretize \([0,T]\) (or the relevant estimation window) on an equally spaced grid \(\{t_k\}_{k=0}^{N_T}\) with step size \(\Delta t\). Throughout the empirical analysis, we set \(\Delta t=1/12\), corresponding to monthly intervals, and allow portfolio rebalancing only at the grid points. For any process \(Y\), we write
\[
\Delta Y(t_k)\coloneqq Y(t_k)-Y(t_{k-1}) 
\]
for its monthly backwards difference. Following Remark~\ref{rmk:dispersion}, we define
\[
\mathrm{RD}(t_k)\coloneqq \mathrm{RD}_{\mathrm{ew}}(t_k;\Delta t) 
\]
as the realized dispersion of the equal-weighted portfolio over the month \([t_{k-1},t_k]\), computed from daily returns. All signals dated \(t_k\) use information available through \(t_k\) and determine the portfolio held over \((t_k,t_{k+1}]\). %
Each monthly interval contains approximately \(252\Delta t=21\) trading days.

To implement the continuous-time results of Sections~\ref{sec:prelim}--\ref{sec:stochastic_control} in the empirical setting, we must account for three departures from the idealized model: (i) portfolios are rebalanced no more frequently than monthly; (ii) instantaneous dispersion is not directly observable, so monthly realized dispersion provides only a proxy; and (iii) changes in the investment universe generate a drag on relative performance through the leakage effect described in Remark~\ref{rmk:leakage}. We incorporate these effects through the following approximation of dispersion:
\[
\delta(t_k)\approx \widehat{\delta}(t_k)\coloneqq c_\delta\,\mathrm{RD}(t_k). 
\]
This approximation is motivated by the empirical observation that, to first order, the local time term in \eqref{eq:leakage} is, when integrated, approximately proportional to cumulative dispersion \(\int_0^t\delta(s)\,ds\); see \cite[Chapters 4--5]{FER02}.
We estimate \(c_\delta\) by ordinary least squares from the regression
\[
\Delta\log V^{\mathrm{ew}}(t_k) - \Delta\varphi_{\mathrm{ew}}(t_k) = \frac{c_\delta}{2}\, \mathrm{RD}_{\mathrm{ew}}(t_k)\Delta t + \epsilon_k. 
\]
Using monthly observations from 1977 to 1994, we obtain \(c_\delta=0.25\).\footnote{See \cite[Section 6.2.1]{CSW25} for an analogous regression.}

\subsection{Calibration and simulation of the mean-reverting SDD model}
\label{sec:log.OU.calibration}

\subsubsection{Calibration}

We calibrate the mean-reverting model to monthly S\&P 500
diversity and realized dispersion data from 1977 through 2004, and
assess its fit over the subsequent sample from 2005 through 2024. Applying
It\^o's formula to the dispersion process in \eqref{eq:logou}, followed
by an Euler--Maruyama discretization with step size \(\Delta t\), gives
\begin{equation}
\label{eq:log.OU.discrete}
\begin{cases}
\displaystyle
\Delta\varphi(t_{k+1})
=
\kappa_\varphi
\bigl(\overline{\varphi}-\varphi(t_k)\bigr)\Delta t
+
\nu_\varphi\sqrt{\delta(t_k)}
\,\Delta W_\varphi(t_{k+1}),
\\[1.2ex]
\displaystyle
\Delta\delta(t_{k+1})
=
\left[
\kappa_\delta
\bigl(\log\overline{\delta}-\log\delta(t_k)\bigr)
+
\frac12\nu_\delta^2
\right]
\delta(t_k)\Delta t
+
\nu_\delta\delta(t_k)
\,\Delta W_\delta(t_{k+1}).
\end{cases}
\end{equation}
Here,
\[
\Delta W_i(t_{k+1})
:=
W_i(t_{k+1})-W_i(t_k),
\qquad i\in\{\varphi,\delta\},
\]
and the Brownian increments satisfy
\[
\E_{t_k}
\left[
\Delta W_\varphi(t_{k+1})
\Delta W_\delta(t_{k+1})
\right]
=
\rho\,\Delta t.
\]
Conditional on the current state, the increment vector in
\eqref{eq:log.OU.discrete} is therefore jointly Gaussian, yielding an
explicit conditional likelihood for the discretely observed state
sequence.

As described in Section~\ref{sec:data}, we treat
\(\widehat{\delta}(t_k)\) as the observed counterpart of
\(\delta(t_k)\) and maximize the resulting conditional likelihood of
the sequence $(\varphi(t_k),\widehat{\delta}(t_k))$, conditional on its initial value.

\begin{table}[ht!]
    \centering
    \begin{tabular}{|l|c|c|}
    \hline
    \textbf{Parameter} & \textbf{Value} & \textbf{95\% CI} \\ \hline
    $\kappa_{\varphi}^{\mathrm{MLE}}$ & 0.322 & (0.108, 0.537) \\ \hline
    $\overline{\varphi}^{\mathrm{MLE}}$ & -6.931 & (-6.990, -6.873) \\ \hline
    $\nu_{\varphi}^{\mathrm{MLE}}$ & 0.343 & (0.317, 0.369) \\ \hline
    $\kappa_{\delta}^{\mathrm{MLE}}$ & 1.895 & (0.888, 2.902) \\ \hline
    $\overline{\delta}^{\mathrm{MLE}}$ & 0.020 & (0.015, 0.026) \\ \hline
    $\nu_{\delta}^{\mathrm{MLE}}$ & 1.306 & (1.207, 1.405) \\ \hline
    $\rho^{\mathrm{MLE}}$ & -0.379 & (-0.471, -0.287) \\ \hline
    \end{tabular}
\caption{Point estimates (via maximum likelihood) and Wald confidence intervals (based on the observed Fisher information) for the discretized mean-reverting model \eqref{eq:log.OU.discrete} fitted by maximum likelihood using monthly diversity and dispersion between 1977--2004.}
    \label{table:sdd_params}
\end{table}

To initialize the numerical maximization of the likelihood, we fit
separate standard \(\mathrm{AR}(1)\) models to \(\varphi\) and
\(\log\delta\), and use the resulting estimates as starting values for
\((\kappa_\varphi,\overline{\varphi},\nu_\varphi)\) and
\((\kappa_\delta,\overline{\delta},\nu_\delta)\), respectively. We
initialize \(\rho\) using the sample correlation between
\(\Delta\varphi\) and \(\Delta\log\mathrm{RD}\) over the estimation
period. The resulting maximum likelihood estimates are reported in
Table~\ref{table:sdd_params}. The estimated long-run diversity level
\(\overline{\varphi}^{\mathrm{MLE}}\) is close to the sample mean of
market diversity over the estimation period; see
Figure~\ref{fig:sdd_simulation}. The 95\% confidence intervals for both
mean-reversion coefficients,
\(\kappa_\varphi^{\mathrm{MLE}}\) and
\(\kappa_\delta^{\mathrm{MLE}}\), exclude zero, providing evidence of
mean reversion in market diversity and log dispersion. The estimate
\(\rho^{\mathrm{MLE}}\) is also significantly negative, indicating a
negative contemporaneous correlation between innovations in diversity
and log dispersion.

\subsubsection{Model fit and stylized properties}

We next assess how well the calibrated mean-reverting model can reproduce several stylized properties of market diversity and dispersion. Figure~\ref{fig:sdd_simulation} displays several simulated paths---generated at a monthly frequency over 48 years---which we compare against S\&P 500 diversity and dispersion.\footnote{All simulated paths are initialized at the January 1977 value of market diversity and the corresponding adjusted realized-dispersion proxy.} When computing estimators of the calibrated SDD model below, we use the blue sample path, which we denote by $(\varphi^\mathrm{SDD}, \delta^\mathrm{SDD})$ to distinguish it from actual diversity and dispersion. We obtain qualitatively similar results using other simulated paths.

Since $\log\delta$ follows an Ornstein–Uhlenbeck process, its stationary distribution is Gaussian with mean $\log\overline{\delta}$ and variance ${\nu_\delta^2}/{2\kappa_\delta}$.
Consequently, the stationary distribution of $\delta$ is $\operatorname{Lognormal} \left( \log\overline{\delta}, \nu_\delta^2/ 2\kappa_\delta \right)$.

\begin{figure}[ht!]
\centering
\includegraphics[scale=0.66]{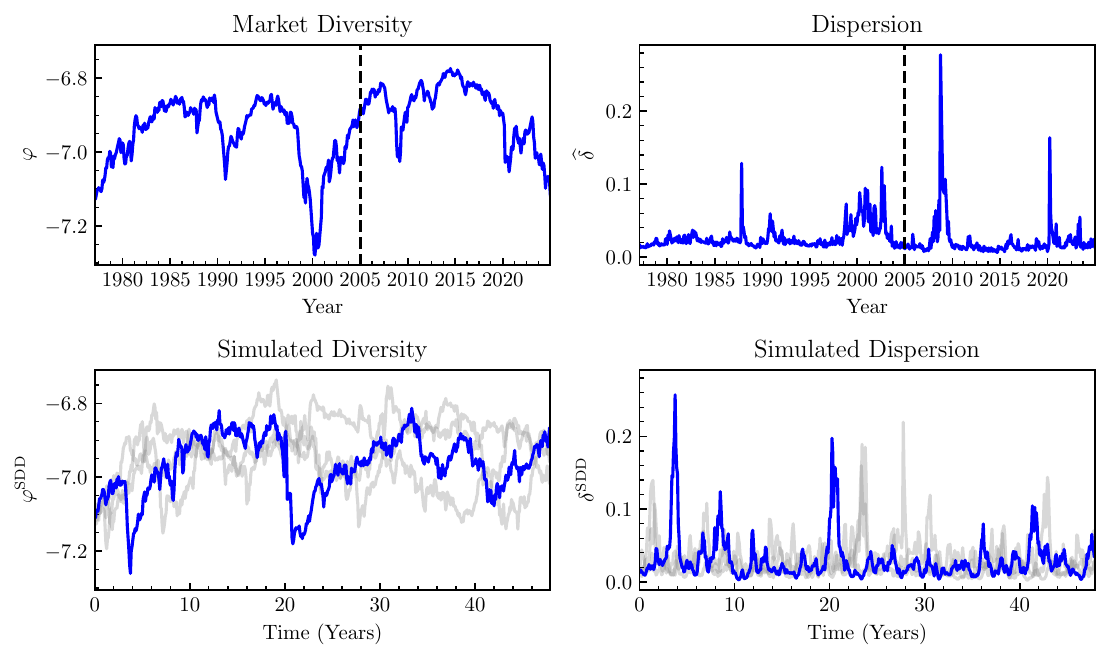}
\caption{Top row: market diversity $\varphi$ and the adjusted dispersion proxy $\widehat{\delta}=c_\delta\mathrm{RD}$ for the S\&P 500 over the full sample period, 1977--2024, both sampled at a monthly frequency. The vertical dashed line in each panel marks the end of the in-sample period used to calibrate the mean-reverting model. Bottom row: monthly sample paths of diversity and dispersion simulated from the calibrated model. The path used in the subsequent distributional comparisons is highlighted in blue, while the remaining paths are shown in grey.}
\label{fig:sdd_simulation}
\end{figure}
We begin by inspecting the empirical distribution of the monthly increments $\Delta \varphi$ of market diversity, plotted on the left panel of Figure~\ref{fig:sdd_kde}. The empirical distribution of the simulated path aligns closely with those arising from both the in- and out-of-sample real data, between which there is very little difference. We observe that the increments are well approximated by a Gaussian distribution.

The model provides a less accurate description of dispersion. Figures~\ref{fig:sdd_simulation} and~\ref{fig:sdd_kde} show that both simulated dynamics and the model-implied stationary distribution meaningfully differ from the empirical data. Although the simulated process exhibits persistent periods of elevated dispersion (so-called volatility clustering), it fluctuates more rapidly at short horizons and does not reproduce the rare extreme spikes observed in the data. The complex multi-scale and nonstationary behaviour of empirical dispersion cannot be faithfully captured by an Ornstein--Uhlenbeck process.%

The left panel of Figure~\ref{fig:sdd_higher_moments} examines the upper tail of dispersion on a logarithmic scale. The in-sample distribution of log dispersion is broadly consistent with a Gaussian approximation and with the calibrated stationary law, whereas the out-of-sample distribution has a substantially heavier right tail.
To quantify this, we fit a Gaussian survival function to the upper tail of each empirical distribution, using approximately the highest 10\% of observations. The fitted standard deviation increases from $0.66$ in sample to $1.46$ out of sample. This difference appears to be driven in large part by extreme observations during the 2008--09 global financial crisis and the 2020 COVID-19 market crash.

The right panel of Figure~\ref{fig:sdd_higher_moments} compares the sample autocorrelation functions of realized dispersion with the stationary autocorrelation function implied by the calibrated model.
For the time lag $h\geq 0$, the latter is
\begin{align*}
\operatorname{Corr}
\bigl(\delta(t),\delta(t-h)\bigr)
= \frac{\exp\left( \frac{\nu_\delta^2}{2\kappa_\delta} e^{-\kappa_\delta h}\right)-1}
{\exp\left(\frac{\nu_\delta^2}{2\kappa_\delta}\right)-1}.
\end{align*}
Thus, for a lag of $\ell$ months, one sets
$h=\ell\Delta t$. The model-implied autocorrelation decays more
rapidly than its in-sample empirical counterpart, but is reasonably
close to the out-of-sample autocorrelation over most of the reported
lags.

\begin{figure}[ht!]
    \centering
    \includegraphics[scale=0.66]
    {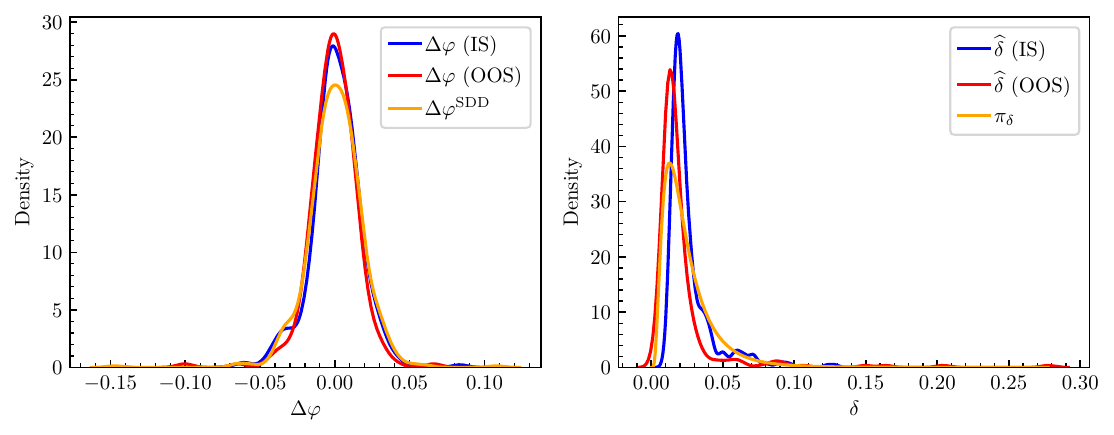}
    \caption{Left panel: Kernel density estimates of the monthly
    increments in market diversity, \(\Delta\varphi\), over the
    in-sample period 1977--2004 (blue), the out-of-sample period
    2005--2024 (red), and the sample path highlighted in
    Figure~\ref{fig:sdd_simulation} (yellow). Right panel: Kernel
    density estimates of the adjusted realized-dispersion proxy
    \(\widehat{\delta}=c_\delta\mathrm{RD}\) over the in-sample and
    out-of-sample periods, together with the stationary density
    \(\pi_\delta\) implied by the calibrated mean-reverting model.}
    \label{fig:sdd_kde}
\end{figure}

\begin{figure}[ht!]
\centering
\includegraphics[scale=0.66]{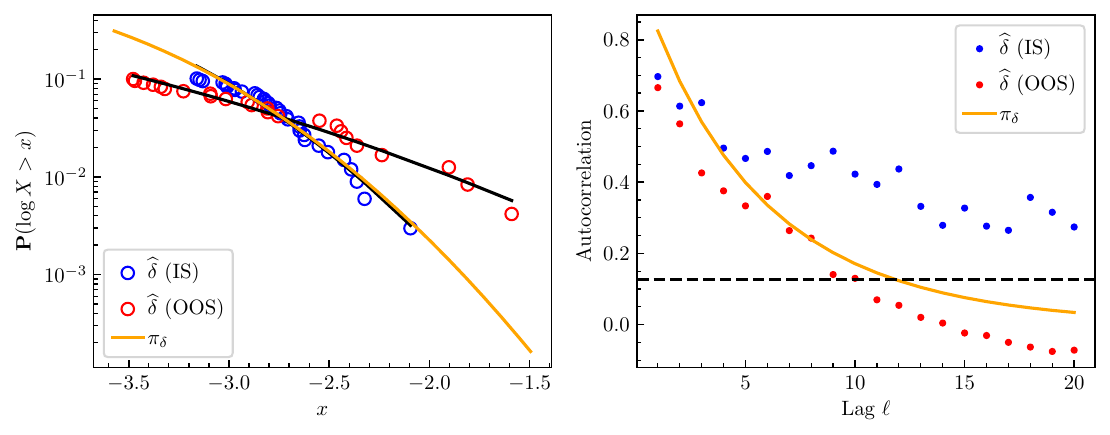}
\caption{Left panel: Empirical upper-tail probabilities of $\log\widehat{\delta}$ over the in-sample and out-of-sample periods, together with the stationary tail implied by the calibrated mean-reverting model. The black curves show Gaussian fits to the upper tail (approximately the highest 10\% of observations) of each empirical distribution, with fitted scale parameters of $0.66$ in-sample and $1.46$ out-of-sample. Right panel: Sample autocorrelation functions of dispersion over the in-sample and out-of-sample periods, together with the stationary autocorrelation function implied by the calibrated model.}
\label{fig:sdd_higher_moments}
\end{figure}

Finally, we evaluate the joint dynamics of market diversity and dispersion under the mean-reverting model. Motivated by \cite[Section 4]{CSW25}, Figure~\ref{fig:sdd_lowess} plots monthly changes in market diversity against corresponding values of monthly realized dispersion, fitted with a locally-weighted scatterplot smoothing (LOWESS) regression. We observe that under the mean-reverting model, consistent with the real data, large changes in diversity exhibit a positive nonlinear association with the level of dispersion.

\begin{figure}[htb!]
\centering
\includegraphics[scale=0.66]{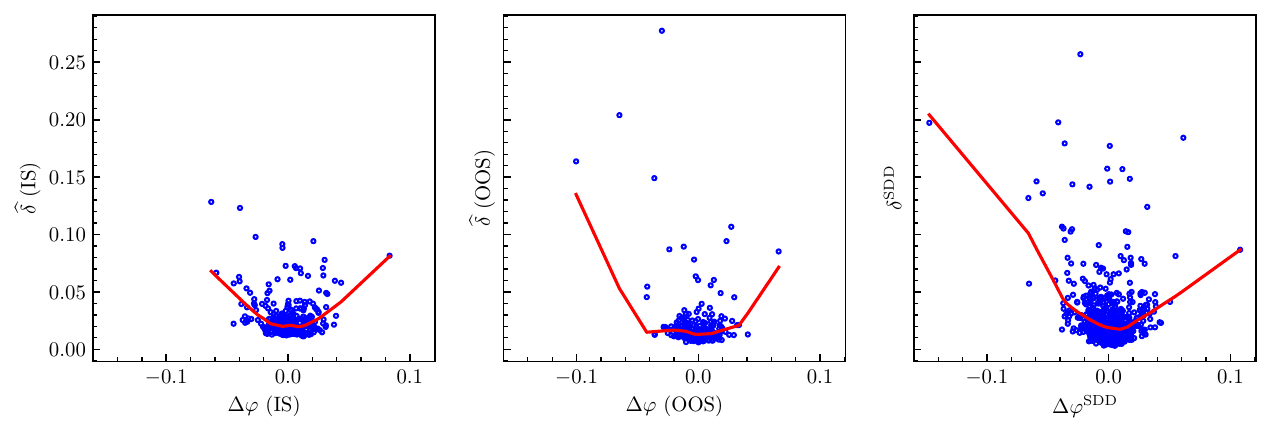}
\caption{Scatterplots of dispersion against contemporaneous monthly increments in market diversity, with LOWESS fits. The left and center panels use the adjusted realized-dispersion proxy over the in-sample and out-of-sample periods, respectively; the right panel uses the highlighted sample path simulated from the calibrated mean-reverting model.}
\label{fig:sdd_lowess}
\end{figure}

\subsubsection{Optimal-control simulations}\label{sec:optimal_control_simulations}
We next use the calibrated mean-reverting model to simulate the
frictionless optimizer \(\lambda_0^*\)\footnote{For $\lambda_0^*$, the adjustment-rate penalty is evaluated using
monthly finite differences.} in
\eqref{eq: frictionless alpha} and the frictional optimizer
\(\lambda^*\) in \eqref{eq: optimal control}, both derived in
Section~\ref{sec:stochastic_control}. We compare the corresponding
portfolios with the equal-weighted portfolio using two cumulative
performance measures, evaluated at the monthly grid points. The first
is the unpenalized log relative value \(\log V^\lambda\). The second is
the penalized log relative value
\[
\log V^\lambda-\mathit{TC}^\lambda,
\]
where \(\mathit{TC}^\lambda\) is the quadratic cost surrogate penalty defined in \eqref{eq:tc_penalty}. \footnote{This quantity evaluates the policy within the quadratic surrogate model. It should not be interpreted as portfolio wealth net of proportional transaction costs. In Section~\ref{sec:backtest}, proportional costs are deducted directly from discretely rebalanced portfolio wealth.}
Both quantities are computed pathwise using the
master formula \eqref{eq:ewmasterformula} with respect to our simulated path $(\varphi^\mathrm{SDD},\delta^\mathrm{SDD})$. The penalized quantity used
in this simulation exercise should not be confused with the net
relative value in Section~\ref{sec:backtest}, where we implement
discrete rebalancing and deduct proportional transaction costs
directly from portfolio wealth.

To fully specify the simulation, we must fix the risk-aversion
parameter \(\gamma\) as well as the cost-penalty coefficients
\(\Lambda_1\) and \(\Lambda_2\). We set $\gamma=20.78$, $\Lambda_1=1.18$ and $\Lambda_2=0.55$, all of which are taken from the calibrated values for the mean-reverting model in the real-data backtest in Section~\ref{sec:backtest}. The value of $\gamma$ is chosen so that the in-sample annualized active volatility of the frictionless optimizer is 5\%. The values of $\Lambda_1$ and $\Lambda_2$ are chosen to reflect proportional transaction costs of 15 bps. Precise calibration details for these parameters can be found in Section~\ref{sssec:backtest design}.

\begin{figure}[ht!]
    \centering
    \includegraphics[scale=0.66]{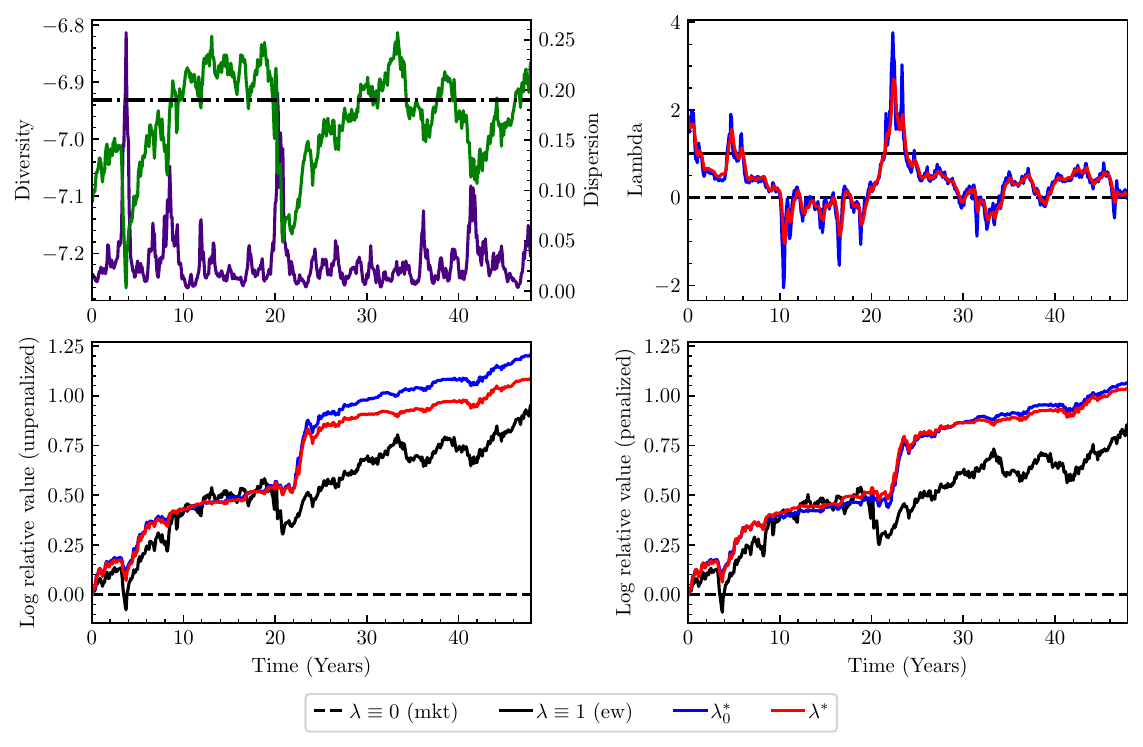}
    \caption{Top row: A simulated path of market diversity and
    dispersion under the calibrated mean-reverting SDD model (left), and the
    resulting frictionless target \(\lambda_0^*\) and frictional
    optimal tilt \(\lambda^*\) (right). Bottom row: The cumulative
    unpenalized log relative values (left) and penalized log relative
    values (right) of the market, equal-weighted, frictionless-optimal,
    and frictional-optimal portfolios. Both optimal strategies use $\gamma=20.78$. The frictional strategy uses
    \(\Lambda_1=1.18\) and \(\Lambda_2=0.55\).}
    \label{fig:sdd_backtest}
\end{figure}

On the simulated path $(\varphi^\mathrm{SDD}, \delta^\mathrm{SDD})$ shown in Figure~\ref{fig:sdd_backtest}, all three
portfolios, $\pi^{\mathrm{ew}}$, $\pi^{\lambda_0^*}$, and
$\pi^{\lambda^*}$, outperform the market in terminal relative-value terms. 
The two optimal tilt portfolios also exhibit smoother relative-value
trajectories than the equal-weighted portfolio. The equal-weighted portfolio
benefits primarily from increases in market diversity and therefore
experiences prolonged relative drawdowns when market concentration rises, as
illustrated by the episode near year 20. By contrast, the optimal tilt
portfolios can take positive or negative exposure to diversity, allowing them
to respond to deviations on either side of its long-run mean. They also adjust their exposure dynamically in response to the conditional drift and
instantaneous variance of diversity, with $\sigma_\varphi^2 \propto \delta$. These adjustments help mitigate periods of underperformance relative to the market and produce less volatile log relative-value trajectories.

Before quadratic implementation-cost penalties are taken into account, the
frictionless portfolio $\pi^{\lambda_0^*}$ attains a higher terminal log
relative value than the frictional portfolio $\pi^{\lambda^*}$ on this
simulated path. This difference reflects the greater responsiveness of the
frictionless optimizer to changes in its target. Such responsiveness,
however, also results in a larger implementation-cost penalty. 
After deducting the surrogate penalties, the two portfolios have similar terminal penalized log relative values on this path.

To assess whether this performance difference extends beyond the particular
simulated path, we simulate 10,000 additional paths of the calibrated
mean-reverting model (over the same time horizon which is 48 years). For each path, we compute the frictionless and
frictional optimizers, $\lambda_0^*$ and $\lambda^*$, and record the difference
in their terminal unpenalized log relative values,
\[
D^{\mathrm{unpen}}
\coloneqq
\log V^{\lambda^*}(T)
-
\log V^{\lambda_0^*}(T),
\]
as well as the corresponding difference after deducting the quadratic
implementation-cost penalty incurred by each strategy,
\[
D^{\mathrm{pen}}
\coloneqq
\left(
\log V^{\lambda^*}(T)
-
\mathit{TC}^{\lambda^*}(T)
\right)
-
\left(
\log V^{\lambda_0^*}(T)
-
\mathit{TC}^{\lambda_0^*}(T)
\right).
\]
Thus, a negative value of $D^{\mathrm{unpen}}$ indicates that the frictionless
portfolio performs better before implementation costs, whereas a positive
value of $D^{\mathrm{pen}}$ indicates that the frictional portfolio performs
better after the respective penalties are deducted.

Figure~\ref{fig:sim_backtest_hist} presents histograms of
$D^{\mathrm{unpen}}$ and $D^{\mathrm{pen}}$, shown in blue and red,
respectively. The unpenalized distribution is centered slightly below zero, whereas
the penalized distribution is centered significantly above zero. Thus, on average, the
frictionless optimizer attains a higher terminal log relative value before
implementation-cost penalties are imposed, whereas the frictional optimizer
attains a higher penalized terminal value once the penalty incurred by each
strategy is deducted.

\begin{figure}[ht!]
\centering
\includegraphics[scale = 0.8]{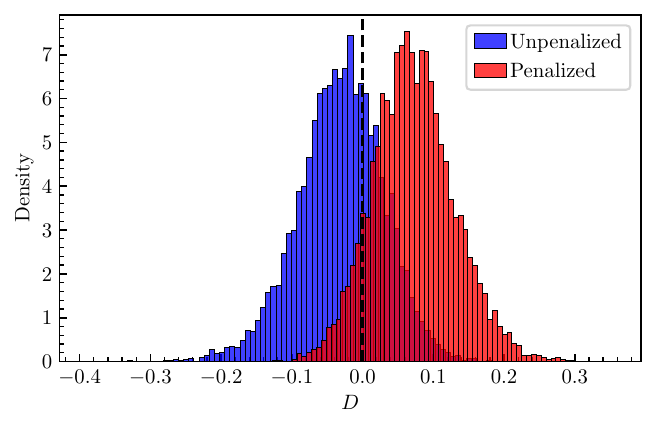}
\caption{Histograms of the terminal performance differences between the frictional and frictionless optimal portfolios across \(10{,}000\) simulated paths of the calibrated mean-reverting SDD model. The blue histogram shows the unpenalized difference \(D^{\mathrm{unpen}}\), while the red histogram shows the surrogate-penalized difference \(D^{\mathrm{pen}}\).}
\label{fig:sim_backtest_hist}
\end{figure}

\subsection{Portfolio implementation and historical backtesting}
\label{sec:backtest}

\subsubsection{Discrete-time implementation}

We now describe the discrete-time implementation of the optimal control
derived in Section~\ref{sec:stochastic_control}. Recall that the
continuous-time solution depends on the terminal horizon $T$ through the
function \(G\). In the empirical implementation, we use a rolling \emph{relative}
control horizon. At each decision time \(t_k\), we set
\[
T_k=t_k+M\Delta t,
\qquad
M\in\mathbb{N},
\]
and define
\[
a
\coloneqq
\sqrt{\frac{1+\gamma+\Lambda_1}{\Lambda_2}},
\qquad
G_k(s)
\coloneqq
\cosh\!\left(a(T_k-s)\right),
\qquad
s\in[t_k,T_k].
\]
At each decision time \(t_k\), we use information available through \(t_k\) to compute the next-period tilt \(\lambda^*(t_{k+1})\). 
Applying \eqref{eq: optimal control} over the interval
\([t_k,t_{k+1}]\) gives
\begin{align*}
\lambda^*(t_{k+1})
&=
\frac{G_k(t_{k+1})}{G_k(t_k)}
\lambda^*(t_k)
-
G_k(t_{k+1})
\int_{t_k}^{t_{k+1}}
\frac{G_k'(s)}{G_k(s)^2}
\lambda^{\mathrm{aim}}(s)\,ds.
\end{align*}
The value \(\lambda^*(t_{k+1})\) is computed using information available at \(t_k\) and applied only to returns realized over \((t_k,t_{k+1}]\). 
Approximating the aim process by its value at the beginning of the
interval yields
\begin{align}
\lambda^*(t_{k+1})
&\approx
\frac{G_k(t_{k+1})}{G_k(t_k)}
\lambda^*(t_k)
+
\left(
1-\frac{G_k(t_{k+1})}{G_k(t_k)}
\right)
\lambda^{\mathrm{aim}}(t_k)
\nonumber\\
&=
\alpha_M\lambda^*(t_k)
+
(1-\alpha_M)\lambda^{\mathrm{aim}}(t_k),
\label{eq:discrete_control}
\end{align}
where
\[
\alpha_M
=
\frac{\cosh\!\left(a(M-1)\Delta t\right)}
     {\cosh\!\left(aM\Delta t\right)}.
\]
Thus, as in the discrete-time framework of
G\^arleanu and Pedersen~\cite{GP13}, the new position is a weighted
average of the inherited position and a forward-looking aim.

For the rolling horizon beginning at \(t_k\), define
\[
K_k(u)
\coloneqq
-\frac{1+\gamma+\Lambda_1}{\Lambda_2}
\frac{G_k(u)}{G_k'(t_k)},
\qquad
u\in[t_k,T_k].
\]
Given the terminal horizon $T_k$, the aim process in \eqref{eq: trading aim} can then be written as
\[
\lambda^{\mathrm{aim}}(t_k)
=
\frac{1+\gamma}{1+\gamma+\Lambda_1}
\int_{t_k}^{T_k}
K_k(u)\,
\E_{t_k}\!\left[\lambda_0^*(u)\right]\,du.
\]
We approximate this integral on the monthly grid by
\begin{align}
\lambda^{\mathrm{aim}}(t_k)
\approx
\frac{1+\gamma}{1+\gamma+\Lambda_1}
\sum_{\ell=0}^{M-1}
\E_{t_k}\!\left[
\lambda_0^*(t_{k+\ell})
\right]
\int_{t_{k+\ell}}^{t_{k+\ell+1}}
K_k(u)\,du.
\label{eq:discrete_aim}
\end{align}
Thus, evaluating \(\lambda^{\mathrm{aim}}(t_k)\) requires the
conditional forecasts
\[
\E_{t_k}\!\left[
\lambda_0^*(t_{k+\ell})
\right],
\qquad
\ell=1,\ldots,M-1,
\]
whereas the term corresponding to \(\ell=0\) is simply the current
frictionless target \(\lambda_0^*(t_k)\). The hyperparameter \(M\)
therefore determines the length \(M\Delta t\) of the rolling forecast
horizon.

When \(M=1\), no forward forecast is required. Since
\(\int_{t_k}^{t_{k+1}}K_k(u)\,du=1\), the aim reduces to
\begin{equation}
\lambda^{\mathrm{aim}}(t_k)
=
\frac{1+\gamma}{1+\gamma+\Lambda_1}
\lambda_0^*(t_k).
\label{eq:M1_aim}
\end{equation}

When backtesting the mean-reverting model below, we set $M=12$, corresponding to a one-year rolling horizon. We find this reasonably balances the tradeoff between reducing transaction costs while minimizing the noise introduced from the forecasting error of $\lambda^\mathrm{aim}$.

\subsubsection{A data-driven trending specification}
\label{sec:trending}

As a data-driven alternative to the mean reverting specification, we also
consider a model designed to capture short-run persistence in changes
in market diversity. Let \(t_k^D\) denote the daily grid and
\(\Delta^D\varphi(t_{k+1}^D)\) the corresponding daily increment in
diversity. We specify
\begin{equation}
\label{eq:trending}
\begin{cases}
\Delta^D\varphi(t_{k+1}^D)
=
m(t_k^D)
+
\epsilon_\varphi(t_{k+1}^D),
\\[0.5ex]
\epsilon_\varphi(t_{k+1}^D)
=
\sqrt{v(t_k^D)}\,z(t_{k+1}^D),
\\[0.5ex]
m(t_{k+1}^D)
=
\alpha_\varphi
\Delta^D\varphi(t_{k+1}^D)
+
(1-\alpha_\varphi)m(t_k^D),
\\[0.5ex]
v(t_{k+1}^D)
=
\omega
+
\alpha_v\epsilon_\varphi(t_{k+1}^D)^2
+
\beta_v v(t_k^D).
\end{cases}
\end{equation}
Here \(z(t_{k+1}^D)\) is a standard Gaussian white-noise process. Thus,
\(m\) is an exponentially smoothed estimate of the local trend in
diversity, whereas \(v\) follows a GARCH(1,1) recursion.

This specification is motivated by empirical evidence of short-run
persistence in changes in market diversity. Audrino et
al.~\cite{AFF07} document positive autocorrelation in
monthly changes in S\&P 500 diversity over 1960--2001. At the same
time, the existence of a fixed long-run equilibrium level for market
diversity is not clear; see \cite[Section~3.2]{CSW25}. The long swings
in diversity visible in Figure~\ref{fig:diversity} also suggest that,
to the extent that diversity is mean reverting, the relevant
adjustment horizon may span several years. The trending specification
is therefore intended as a flexible empirical alternative to the
long-run mean-reversion dynamics of the mean-reverting model.

Using the recursions in \eqref{eq:trending}, the conditional mean and
variance of the next monthly change in diversity are
\[
b_\varphi(t_k)
=
\frac{
\E_{t_k}[\Delta\varphi(t_{k+1})]
}{
\Delta t
},
\qquad
\sigma_\varphi^2(t_k)
=
\frac{
\Var_{t_k}[\Delta\varphi(t_{k+1})]
}{
\Delta t
}.
\]
Since the conditional expectation of the smoothed daily trend remains
equal to its current value over the forecast horizon,
\[
b_\varphi(t_k)
=
\frac{m(t_k)}{\Delta t^D}.
\]
For a month consisting of 21 trading days, the conditional variance is
\begin{equation}\label{eq:conditional_variance}
\sigma_\varphi^2(t_k)
=
\frac{1}{\Delta t}
\sum_{i=1}^{21}
\left(
1+\alpha_\varphi(21-i)
\right)^2
\left[
\overline v
+
(\alpha_v+\beta_v)^{i-1}
\left(
v(t_k)-\overline v
\right)
\right],
\end{equation}
where $\overline v = \omega / (1-\alpha_v-\beta_v)$ is the stationary mean of the conditional daily variance, provided
\(\alpha_v+\beta_v<1\).

For the trending specification we set $M=1$, so that the aim process $\lambda^\mathrm{aim}$ reduces to Equation~\eqref{eq:M1_aim}. No dynamic model for future dispersion is, therefore, required. At each
decision time, we use the contemporaneous adjusted dispersion proxy $\widehat{\delta}(t_k) = c_\delta\mathrm{RD}(t_k)$ as discussed in Section~\ref{sec:data}.

\subsubsection{Backtest design}\label{sssec:backtest design}

We evaluate all strategies over the period 1995--2024, using the
preceding 18 years, 1977--1994, as the in-sample period for parameter
estimation and calibration. Each portfolio is initialized with wealth
\(z_0=\$1{,}000\) and rebalanced monthly, with each target portfolio held over the subsequent month. 
We assume that stocks are held throughout the entire month unless they delist from their respective stock exchange (\textit{not} the S\&P 500 index) prior to the end of the month. Portfolios are initialized without incurring transaction costs. Following \cite{TM21}, we impose proportional transaction costs of $\mathrm{tc}=15$ basis points (bps) on the value of all trades by adapting the methodology of Ruf \& Xie (\cite{RX20}) to allow for short-selling. Following the notation in \cite{RX20}, given a portfolio $\pi$, let $\psi \coloneqq Z^\pi \pi$ denote the dollar value of its positions, let $t_k-$ denote the time just before rebalancing at time $t_k$, and let $D(t_k-)$ represent the net dividends received/paid during the period $[t_{k-1},t_k)$. We can express the self-financing condition via the equation \begin{equation}\label{eq:self-financing}
0 = \sum\limits_{i=1}^{n} (\psi_i(t_k)-\psi_i(t_{k}-)) + \mathrm{tc} \sum\limits_{i=1}^{n} |\psi_i(t_k)-\psi_i(t_{k}-)| - D(t_k-),
\end{equation}
where the first summand represents the dollar value of trades made at time $t_k$ and the second summand represents the cumulative transaction costs incurred. Defining $c \coloneqq Z^\pi(t_k) / Z^\pi(t_k-)$ and $c_i\coloneqq \pi_i(t_k-) / \pi_i(t_k) \mathbf{1}_{\pi_i(t_k)\neq 0}$, after dividing out by $Z^\pi(t_k-)$, we can rewrite Equation~\eqref{eq:self-financing} as $F(c)=0$, where \[
F(c) \coloneqq \sum_{i=1}^n  \pi_i(t_k)(c-c_i) + \mathrm{tc} \sum_{i=1}^n |\pi_i(t_k)(c-c_i)| - \widehat{D}(t_k-)
\] with \[\widehat{D}(t_k-) \coloneqq \frac{D(t_k-) + \sum\limits_{i=1}^{n} (\psi_i(t_k-)- \mathrm{tc}|\psi_i(t_k-)| )\mathbf{1}_{\pi_i(t_k)=0}}{Z^{\pi}(t_k-)}.\] As $c \to \pm \infty$, $F(c)\to \pm \infty$, and the first-order conditions imply $F$ is monotone increasing, and hence admits a unique root, provided $\mathrm{tc}\|\pi(t_k)\|_1<1$. This condition always holds for reasonable portfolios and proportional transaction costs.

We also follow \cite{RX20} in defining the dividend rate $R^{D}_i(t_k)$ for the $i^{th}$ stock over $[t_{k-1}, t_{k})$ by \[
R^{D}_{i}(t_k) \coloneqq \max  \left(1+R_i(t_k)-\frac{S_i(t_{k})}{S_{i}(t_{k-1})}, 0 \right) ,
\] where the simple return $R_i(t_k)$ over $[t_{k-1},t_k]$ includes dividend re-investments. Using the change in market capitalization ensures the dividend rate contains not only the actual stock dividends, but also the effects of other corporate actions.

We compare the market and equal-weighted portfolios with the strategies
generated by the mean reverting and trending SDD specifications. For each model
\(m\in\{\mathrm{mr},\mathrm{trend}\}\), we consider the frictionless
optimizer \(\lambda_0^{*,m}\), the frictional optimizer
\(\lambda^{*,m}\), and the clipped strategy
\[
\overline{\lambda}^{*,m}
\coloneqq
\bigl(\lambda^{*,m}\vee0\bigr)\wedge1.
\]
Since \(0\leq\overline{\lambda}^{*,m}\leq1\), the corresponding
portfolio is a convex combination of the market and equal-weighted
portfolios and is therefore long-only. We emphasize that this clipped
strategy is included as an implementable benchmark; it is not the
solution of a dynamically constrained version of our stochastic
control problem.

Because the backtest uses a shorter estimation window than the calibration
exercise in Section~\ref{sec:log.OU.calibration}, we re-estimate the
parameters of the mean-reverting SDD model using data from 1977--1994 and
hold them fixed throughout the out-of-sample period. Although the state
variables and conditional forecasts are updated as new observations arrive,
the model parameters are \emph{not} re-estimated out of sample. The resulting
parameter estimates are reported in Table~\ref{table:backtest_params}.
Relative to those obtained from the longer 1977--2004 calibration sample,
the estimates of $\nu_\varphi$, $\kappa_\delta$, and $\nu_\delta$ differ
materially. These differences may partly reflect the inclusion of the
elevated dispersion surrounding the dot-com episode in the longer
calibration sample, which is absent from the shorter backtest-estimation
window.

For the trending specification, we similarly estimate the GARCH
parameters from in-sample daily changes in diversity. As we see in Table~\ref{table:backtest_params}, the persistence $\alpha_v+\beta_v=0.94$ is less than one, corresponding to stationary volatility. We choose the
smoothing parameter \(\alpha_\varphi\) so that
\[
h_\varphi
\coloneqq
-\frac{\log 2}{\log(1-\alpha_\varphi)}
=
21
\]
trading days, corresponding approximately to one month and consistent
with the persistence documented in \cite{AFF07}. Increasing the
half-life to several months produces qualitatively similar backtest
results, consistent with the persistence horizons commonly observed
for time-series momentum strategies \cite{MOP12}.

To facilitate comparison across forecasting specifications, we
risk-normalize the two strategies in sample. Specifically, for each
model \(m\), we choose the risk-aversion parameter \(\gamma_m\) so
that the frictionless target portfolio has annualized active volatility
of 5\%:
\begin{equation}
\sqrt{
\frac{1}{\Delta t}
\widehat{\Var}_{\mathrm{IS}}
\left(
r^{\lambda_0^{*,m},\mathrm{gross}}
-
r^{\mu,\mathrm{gross}}
\right)
}
=
5\%.
\label{eq:risk_target}
\end{equation}
Here and below, ``gross'' refers to portfolio performance before
proportional transaction costs, whereas ``net'' refers to performance
after deducting those costs. The model-specific values of \(\gamma_m\)
therefore normalize the two forecasting specifications to a common
active-risk budget.

\begin{table}[ht!]
    \centering
\footnotesize
\setlength{\tabcolsep}{3pt}
    \begin{tabular}{|l|c|c|c|c|c|c|c|c|c|c|c|c|}
    \hline
    \textbf{SDD Model} & \multicolumn{8}{c|}{Mean-reverting} & \multicolumn{4}{c|}{Trending}\\
    \hline
    \textbf{Parameter} & $\kappa_\varphi$ & $\overline{\varphi}$ & $\nu_\varphi$ & $\kappa_\delta$ & $\overline{\delta}$ & $\nu_\delta$ & $\rho$ & $\gamma_\mathrm{mr}$  & $\omega$ & $\alpha_v$ & $\beta_v$  & $\gamma_\mathrm{trend}$ \\ \hline 
    \textbf{Value} & 0.43 & $-6.93$ & 0.27 & 5.50 & 0.021 & 1.00 & $-0.30$ & 20.78  & $5.53 \cdot 10^{-7}$ & 0.14 & 0.80 & 29.25 \\ \hline 
    \end{tabular} 
\caption{Model parameters for the mean-reverting and trending SDD strategies estimated in-sample between 1977--1994.}
    \label{table:backtest_params}
\end{table}

\paragraph{Calibration of the quadratic surrogate to proportional costs}
We now map the proportional-cost implementation into the quadratic surrogate used to construct the policy. The calibration target is the cumulative active wealth drag generated by proportional transaction costs, rather than the instantaneous cost of an individual trade. This distinction is essential: proportional costs are linear in absolute turnover, whereas the surrogate is quadratic in the level and adjustment speed of the tilt. We choose the surrogate parameters using only the in-sample period and subsequently evaluate the resulting policy by deducting proportional costs directly from out-of-sample wealth.

For \(\bullet\in\{\mathrm{gross},\mathrm{net}\}\), let
\[
V^{\pi,\bullet}(t)
=
\frac{Z^{\pi,\bullet}(t)}
     {Z^{\mu,\bullet}(t)}
\]
denote the relative value of portfolio \(\pi\) with respect to the
correspondingly implemented market portfolio. We define the cumulative
active cost of \(\pi\) by
\begin{equation}
C^\pi(t)
\coloneqq
\log V^{\pi,\mathrm{gross}}(t)
-
\log V^{\pi,\mathrm{net}}(t),
\label{eq:active_cost}
\end{equation}
and write \(C^\lambda\coloneqq C^{\pi^\lambda}\) for a
\(\lambda\)-tilt portfolio.
Thus, $C^\pi$ is the empirical calibration target generated by the proportional-cost backtest, whereas $TC^\lambda$ is the analytically tractable surrogate appearing in the control problem.

Because \(C^\pi\) is defined in terms of relative rather than absolute
portfolio values, it need not coincide with the logarithm of the
cumulative dollar transaction costs incurred by \(\pi\). In
particular, the market benchmark itself incurs some turnover in the
empirical implementation, for example when the S\&P 500 constituent
set changes. Consequently, increments of \(C^\pi\) need not be
nonnegative in general, although they are positive throughout our
monthly backtests.

Recall that the quadratic cost proxy associated with a tilt process is
\[
\mathit{TC}^{\lambda}(t)
=
\frac12
\int_0^t
\left(
\Lambda_1\lambda(s)^2
+
\Lambda_2\dot\lambda(s)^2
\right)
\sigma_\varphi^2(s)\,ds.
\]
The equal-weighted portfolio provides a natural anchor for $\Lambda_1$: corresponding to the $\lambda$-tilt portfolio with $\lambda\equiv1$ and $\dot\lambda\equiv0$, its surrogate penalty isolates the maintenance component:
\[
\mathit{TC}^{\mathrm{ew}}(t)
=
\frac{\Lambda_1}{2}
\int_0^t
\sigma_\varphi^2(s)\,ds.
\]
For each forecasting specification \(m\), we therefore calibrate
\(\Lambda_1^m\) by matching the terminal in-sample quadratic cost proxy
to the observed active cost of the equal-weighted portfolio:
\begin{equation}
\widehat{\Lambda}_1^m
=
\frac{
2C^{\mathrm{ew}}(T_{\mathrm{IS}})
}{
\displaystyle
\int_0^{T_{\mathrm{IS}}}
\sigma_{\varphi,m}^2(s)\,ds
},
\label{eq:lambda1_calibration}
\end{equation}
where $[0, T_{\mathrm{IS}}]$ is the in-sample period (1977--1994). Equation~\eqref{eq:lambda1_calibration} therefore matches the terminal in-sample wealth drag of the equal-weighted portfolio under proportional costs. 
Under the mean reverting specification with
\(\sigma_{\varphi,\mathrm{mr}}^2
=\nu_\varphi^2\delta\),
\eqref{eq:lambda1_calibration} becomes
\[
\widehat{\Lambda}_1^{\mathrm{mr}}
=
\frac{
2C^{\mathrm{ew}}(T_{\mathrm{IS}})
}{
\displaystyle
\nu_\varphi^2
\int_0^{T_{\mathrm{IS}}}
\delta(s)\,ds
}.
\]
The trending specification uses the GARCH forecasts of $\sigma^2_{\varphi,m}$ given by Equation~\eqref{eq:conditional_variance}. For the equal-weighted portfolio itself, we estimate the denominator of Equation~\eqref{eq:lambda1_calibration} using the monthly realized quadratic variation of diversity. 

Having fixed \(\widehat{\Lambda}_1^m\), we calibrate
\(\Lambda_2^m\) using the corresponding frictional optimal strategy.
More precisely, we choose \(\widehat{\Lambda}_2^m\) to minimize the
absolute terminal in-sample discrepancy between the quadratic cost
proxy and the realized active cost:
\begin{equation}
\widehat{\Lambda}_2^m
\in
\argmin_{\Lambda_2>0}
\left|
\mathit{TC}^{\lambda^{*,m}(\Lambda_2)}
(T_{\mathrm{IS}})
-
C^{\lambda^{*,m}(\Lambda_2)}
(T_{\mathrm{IS}})
\right|,
\label{eq:lambda2_calibration}
\end{equation}
with all other parameters held fixed. 
Conditional on $\widehat\Lambda_1^m$, the dynamically adjusted strategy identifies $\Lambda_2^m$, since its implementation generates both constituent-level rebalancing and changes in the tilt. Equations~\eqref{eq:lambda1_calibration}--\eqref{eq:lambda2_calibration} therefore define a two-stage effect calibration from the proportional-cost model to the quadratic surrogate. They should not be interpreted as structural estimation of the proportional cost function. All parameters are estimated using 1977--1994 data and then held fixed throughout the out-of-sample period.

\begin{figure}[ht!]
    \centering
    \includegraphics[scale=0.66]{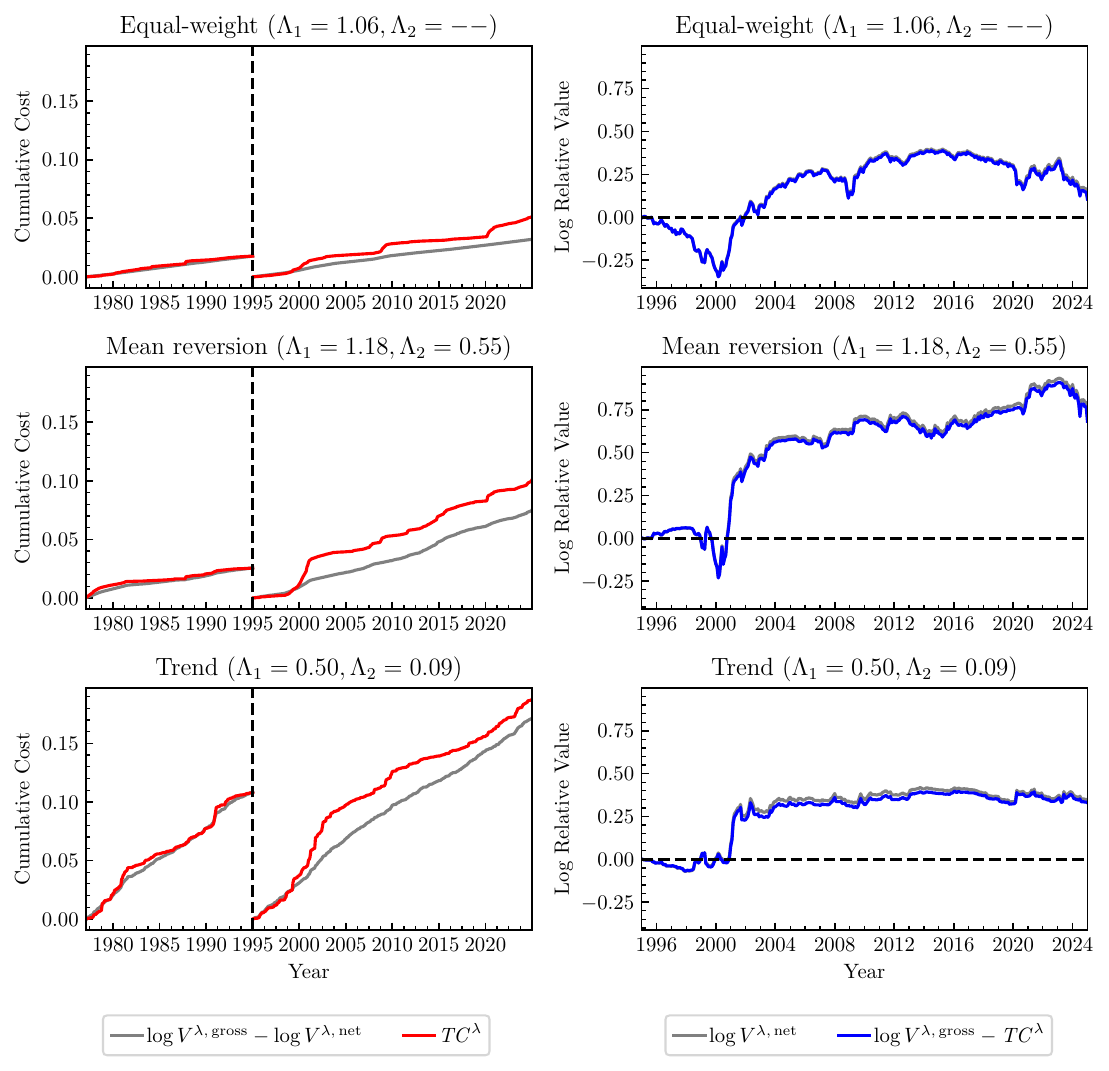}
    \caption{Left column: Cumulative active costs
    \(C^\lambda
    =
    \log V^{\lambda,\mathrm{gross}}
    -
    \log V^{\lambda,\mathrm{net}}\)
    generated by monthly rebalancing under 15 bps proportional
    transaction costs (gray), together with the calibrated quadratic
    cost proxy \(\mathit{TC}^\lambda\) (red). The vertical dashed
    lines mark the end of the in-sample period in December 1994.
    Right column: Actual net log relative value
    \(\log V^{\lambda,\mathrm{net}}\) (gray) and its quadratic-penalty
    approximation
    \(\log V^{\lambda,\mathrm{gross}}
    -\mathit{TC}^\lambda\) (blue).}
    \label{fig:calibrating_costs}
\end{figure}

Figure~\ref{fig:calibrating_costs}  assesses the effect calibration beyond the terminal in-sample matching conditions. The vertical dashed line separates the calibration and evaluation periods. Although the surrogate parameters are estimated using data only through December 1994, the quadratic penalty broadly tracks the subsequent accumulation of active wealth drag under proportional costs. The fit is not exact, nor is exact equality required for the policy-design approximation. The relevant question is whether the surrogate captures the economically important magnitude and timing of implementation frictions sufficiently well to generate a useful smoothed policy.

\subsubsection{Backtest results}

The quadratic surrogate affects the strategy through its calibrated policy parameters, but it is not used to compute the reported net returns. Every net-return, Sharpe-ratio, information-ratio, drawdown, and wealth statistic in this subsection is computed after deducting 15-bps proportional transaction costs directly from portfolio wealth.

Table~\ref{table:backtest} summarizes the out-of-sample performance
of the strategies. Both the mean-reverting and trending specifications produce optimal
portfolios that outperform the market and equal-weighted portfolios
over the 1995--2024 period. Among the frictional strategies, the
mean-reverting specification delivers the stronger relative performance, with
an information ratio of \(0.35\), compared with \(0.29\) for the
trending specification. The equal-weighted portfolio, by contrast,
has an information ratio of \(0.13\) over the same period. On an absolute basis,
we see the average rate of return and volatility of all of the portfolios
are quite similar, resulting in Sharpe ratios which are not materially 
different from each other. This is expected, as the stochastic control 
framework we considered in Section~\ref{sec:stochastic_control} formalized
an \emph{active} portfolio management problem in which we optimized performance
relative to a market benchmark.

\begin{table}[t!]
\centering
\footnotesize
\begin{tabular}{|l|c|c|c|c|c|c|c|c|c|c|}
\hline
Portfolio
& Tot. Ret.
& Ret.
& Vol.
& SR\textsubscript{gross}
& SR\textsubscript{net}
& IR
& Turnover
& GE
& MDD
& RT
\\
\hline
$\lambda^{*,\mathrm{mr}}_0$
& 31.75 & 0.13 & 0.17 & 0.69 & 0.65 & 0.20
& 0.40 & 1.61 & 0.51 & 52
\\
$\lambda^{*,\mathrm{mr}}$
& \textbf{44.38} & \textbf{0.14} & 0.17 & 0.74 & \textbf{0.72} & \textbf{0.35}
& 0.15 & 1.43 & \textbf{0.50} & 50
\\
$\overline{\lambda}^{*,\mathrm{mr}}$
& 32.34 & 0.13 & 0.16 & 0.68 & 0.68 & 0.34
& 0.05 & \textbf{1.00} & 0.51 & \textbf{39}
\\
$\lambda^{*,\mathrm{trend}}_0$
& 31.35 & 0.13 & 0.16 & \textbf{0.75} & 0.68 & 0.25
& 0.57 & 1.16 & 0.52 & 41
\\
$\lambda^{*,\mathrm{trend}}$
& 30.60 & 0.13 & \textbf{0.15} & 0.72 & 0.68 & 0.29
& 0.33 & 1.10 & 0.52 & 52
\\
$\overline{\lambda}^{*,\mathrm{trend}}$
& 26.48 & 0.12 & 0.16 & 0.66 & 0.65 & 0.29
& 0.15 & \textbf{1.00} & 0.52 & 41
\\
$\pi^{\mathrm{ew}}$ $(\lambda\equiv1)$
& 24.33 & 0.12 & 0.17 & 0.59 & 0.58 & 0.13
& 0.07 & \textbf{1.00} & 0.56 & 44
\\
$\mu$ $(\lambda\equiv0)$
& 21.34 & 0.12 & \textbf{0.15} & 0.61 & 0.61 & ---
& \textbf{0.01} & \textbf{1.00} & \textbf{0.50} & 73
\\
\hline
\end{tabular}
\caption{\small
Summary of portfolio performance under monthly rebalancing and
15 bps proportional transaction costs over the out-of-sample period
1995--2024. Tot.\ Ret.\ denotes cumulative net return, $Z^{\pi, \mathrm{net}}(T) / Z^{\pi, \mathrm{net}}(0) - 1$. Ret.\ is the
annualized mean net return,
\(\overline{R}^{\,\pi,\mathrm{net}}/\Delta t\), and Vol.\ is the
annualized volatility,
\(\operatorname{sd}(R^{\pi,\mathrm{net}})/\sqrt{\Delta t}\).
SR\textsubscript{gross} and SR\textsubscript{net} denote the gross
and net Sharpe ratios, respectively, using the risk-free rate from
Ken French's data library \cite{FRE}. IR denotes the information
ratio of the net active return
\(R^{\pi,\mathrm{net}}-R^{\mu,\mathrm{net}}\).
Turnover is the average \(\ell_1\)-distance between the pre- and
post-rebalancing portfolio weights,
\(\lVert\pi(t_k)-\pi(t_k-)\rVert_1\).
GE denotes average gross exposure,
\(\lVert\pi(t_k)\rVert_1\).
MDD and RT denote maximum drawdown and recovery time, respectively,
with RT measured in months.}
\label{table:backtest}
\end{table}

The differences between the two forecasting specifications can be
understood from the trajectories in
Figure~\ref{fig:alpha_port_perf}. The trend strategy responds to large, 
sudden changes in diversity, preempting the persistence of those changes
into the future. This mitigates relative drawdowns during episodes in which
diversity suddenly drops significantly, such as the dot-com bubble and 
the global financial crisis. The strategy can also take a negative
equal-weighted tilt when the estimated trend is sufficiently adverse,
as occurs around the COVID-19 market disruption.

The mean-reverting strategy instead responds primarily to the level of
diversity relative to its estimated long-run mean. Its target therefore
varies more slowly and reflects a longer horizon mean-reversion mechanism.
This distinction is particularly visible during the most recent years
of the sample, when diversity remains persistently below its historical
mean and the two forecasting specifications generate materially
different positions. The recent relative underperformance of the mean-reverting
strategy is particularly interesting. It is unclear if, as in the lead up 
to the dot-com crash, this represents a decline in market diversity which 
is unsustainable, or else a structural break in the mean-reverting dynamics of diversity.

\begin{figure}[ht!]
    \centering
    \includegraphics[scale=0.66]
    {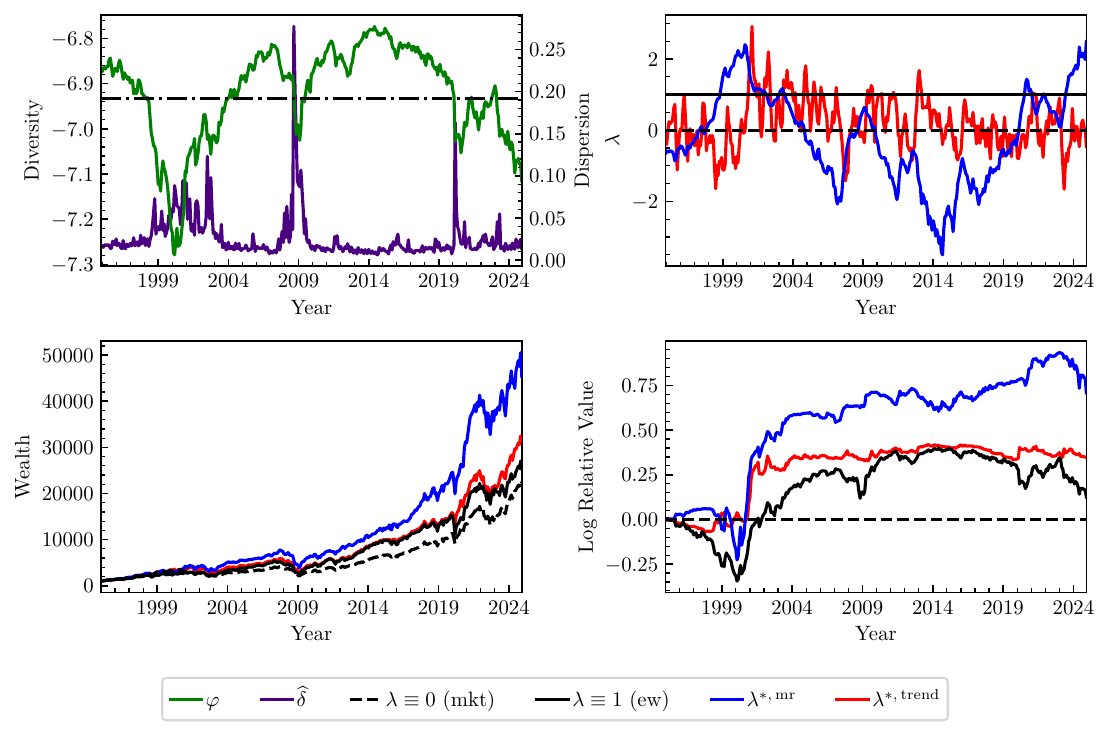}
    \caption{Top row: Historical S\&P 500 diversity and dispersion
    (left), and the optimal frictional tilt \(\lambda^*\) under the
    mean-reverting and trending specifications (right), together with the
    reference levels \(\lambda\equiv0\) (market) and
    \(\lambda\equiv1\) (equal weight). Bottom row: Corresponding
    out-of-sample portfolio wealth (left) and net log relative value
    with respect to the market (right), after 15 bps proportional
    transaction costs.}
    \label{fig:alpha_port_perf}
\end{figure}

The trending strategies exhibit substantially greater turnover than
their mean-reverting counterparts. This reflects the different time scales of
the two forecasting specifications: the trending model places greater
weight on recent changes in diversity, whereas the mean-reverting model
generates a more slowly varying target through long-run mean reversion.
For both specifications, moving from the frictionless target
\(\lambda_0^*\) to the frictional optimizer \(\lambda^*\) reduces
turnover while improving the net information ratio. Under the
mean-reverting specification, turnover falls from \(0.40\) to \(0.15\) and the
information ratio rises from \(0.20\) to \(0.35\); under the trending
specification, turnover falls from \(0.57\) to \(0.33\) and the
information ratio rises from \(0.25\) to \(0.29\). For the mean-reverting model, the significant improvement in relative performance associated with the frictional optimizer cannot be entirely explained by a reduction in transaction costs. Rather, a slower-varying trading signal appears to better predict the mean-reverting dynamics of market diversity. This is further corroborated by the findings in the sensitivity analysis carried out in Section~\ref{sec:sensitivity_analysis}.

The mean-reverting strategy tends to take larger short positions, with an average gross exposure of \(1.43\), compared with
\(1.10\) for the trending strategy, approximately \(29\%\) higher. Clipping the tilt process to the interval \([0,1]\) leaves the information ratio of both strategies effectively unchanged, indicating their respective outperformance does not require excessive short-selling.

\subsubsection{Sensitivity analysis}\label{sec:sensitivity_analysis}

We finally examine sensitivity to the penalty parameters of the quadratic surrogate penalty.
For each model
\(m\in\{\mathrm{mr},\mathrm{trend}\}\), we perturb the calibrated
coefficients according to
\[
\Lambda_1^m
=
\widehat{\Lambda}_1^m 2^{N_1},
\qquad
\Lambda_2^m
=
\widehat{\Lambda}_2^m 2^{N_2},
\]
while holding the forecasting-model parameters and the risk-aversion
coefficient \(\gamma_m\) fixed. Figure~\ref{fig:ir_heatmap} reports
the resulting out-of-sample information ratios.

For both models, performance is relatively
insensitive to moderate changes in \(\Lambda_1\). This is consistent
with the structure of the optimal control, in which \(\gamma\) and
\(\Lambda_1\) enter the adjustment dynamics through the combination
\(1+\gamma+\Lambda_1\); around the calibrated values, moderate
perturbations of \(\Lambda_1\) therefore have only a limited effect on
the policy.

The two strategies differ more noticeably in their sensitivity to \(\Lambda_2\). The mean-reverting strategy remains relatively stable over a
broad range of trading-rate penalties. By contrast, sufficiently large values of \(\Lambda_2\) materially reduce the performance of the trending strategy and can eventually produce a negative information ratio. This difference is economically natural: the trending specification is designed to exploit short-run persistence and therefore requires greater responsiveness to changes in the estimated state, whereas the mean-reverting strategy is driven by more slowly varying mean-reversion dynamics. Excessive smoothing consequently destroys more of the forecasting value of the trending strategy.

\begin{figure}[ht!]
    \centering
    \includegraphics[scale=0.66]{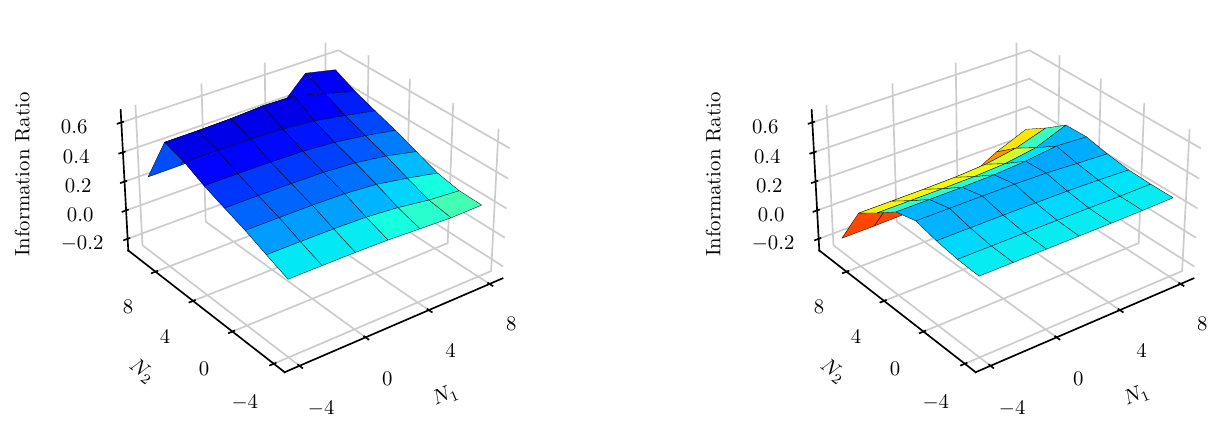}
    \caption{Out-of-sample information ratio of the frictional
    optimal portfolio under the mean-reverting specification (left) and the
    trending specification (right) as the quadratic cost coefficients
    are varied according to
    \(\Lambda_1=\widehat{\Lambda}_1 2^{N_1}\) and
    \(\Lambda_2=\widehat{\Lambda}_2 2^{N_2}\), relative to their
    calibrated baseline values. All other parameters are held fixed.}
    \label{fig:ir_heatmap}
\end{figure}

\section{Conclusion}
We develop a stochastic control framework for dynamically allocating between the market and equal-weighted portfolios, with a natural extension to functionally generated portfolios. Stochastic portfolio theory reduces the relative performance of these allocations to the dynamics of market diversity and dispersion, motivating our low-dimensional stochastic diversity--dispersion (SDD) model. Within this framework, the investor balances expected relative performance and active risk against an effect-calibrated quadratic surrogate for implementation frictions. The resulting control problem admits a linear forward--backward SDE representation and an explicit optimal trading rule.

For a mean-reverting SDD specification, the conditional forecasts entering the optimal strategy are available explicitly. Using historical S\&P 500 data, we calibrate the surrogate parameters in sample to the cumulative wealth effect of proportional transaction costs and evaluate out-of-sample performance by deducting those costs directly from portfolio wealth. The resulting strategy outperforms both the market and equal-weighted portfolios over the 1995--2024 evaluation period. A data-driven trending specification produces similar improvements and illustrates the flexibility of the framework.

Natural extensions include richer models of diversity and dispersion incorporating macroeconomic predictors or additional state variables, as in \cite{AFF07,KIM23}, and extensions to open markets that model constituent turnover and leakage explicitly, following Remark~\ref{rmk:leakage} and \cite{KK21}. Another direction is to incorporate proportional transaction costs or price impact directly into the wealth dynamics, as in \cite{I25}, rather than through the tractable quadratic surrogate used here.

\appendix
\section{Proofs} \label{sec:proofs}
\subsection{Proof of Theorem~\ref{thm: FBSDE}}\label{app:thm proof}

\begin{proof}
We present a self-contained proof to ensure that the admissible requirements in our setting are satisfied. 
Since the objective functional is strictly concave over the admissible strategy set when $\Lambda_2>0$ and $1+\gamma+\Lambda_1>0$, its first-order condition is sufficient for global optimality.

Firstly, via direct calculation, we can see that $\lambda(t)$ given by~\eqref{eq: optimal control} and $\dot\lambda(t)$ given by~\eqref{eq: optimal trading} satisfy the FBSDE system~\eqref{fbsde:constant}, establishing existence. 

Second, it is straightforward to see that the optimal combination ratio $\lambda(t)$ given by~\eqref{eq: optimal control} belongs to $ \mathbf{L}_{\mathbb{F}}^{4}(T;\mathbb{R})\subset\mathbf{L}_{\mathbb{F}}^{2}(T;\mathbb{R})$. 
Since 
\[
\dot\lambda(t) = \frac{1+\gamma+\Lambda_1}{\Lambda_2}\E_t\left[\int_t^T\left(\frac{1+\gamma}{1+\gamma+\Lambda_1}\lambda_0^*(u)-\lambda(u)\right)du\right],
\]
we can also infer that $\dot\lambda\in \mathbf{L}_{\mathbb{F}}^{4}(T;\mathbb{R})\subset\mathbf{L}_{\mathbb{F}}^{2}(T;\mathbb{R})$. 
Further, given the integrability condition of $b_\varphi(t), \sigma_\varphi(t)$ and $\delta(t)$ by \eqref{cond:integrability}, 
with constants $C_1,C_2>0$ that depend only on $\gamma, \Lambda_1, \Lambda_2,T$,
we can further infer that 
\begin{align*}
\E\left[\int_0^T (\sigma_\varphi(t) \lambda(t))^2 dt\right]
&\leq\left(\E\left[\int_0^T \sigma_\varphi^4(t) dt\right]\E\left[\int_0^T \lambda^4(t) dt \right]\right)^{1/2}
\\&\leq \left(\E\left[\int_0^T \sigma_\varphi^4(t) dt\right]\right)^{1/2}\left(C_1 +  C_2\E\left[\int_0^T (\lambda_0^*(t))^4 dt\right]\right)^{1/2}
<\infty. 
\end{align*}
Moreover, we can infer that
\begin{align*}
\E\left[\int_0^T (\sigma_\varphi(t) \dot\lambda(t))^2  dt\right]
\leq \frac{1}{2} \E\left[\int_0^T\sigma_\varphi^4(t) dt + \int_0^T \dot\lambda^4(t) dt \right] <\infty,
\end{align*}
which satisfies the integrability condition~\eqref{frictional cond: integrability}.

Now we would like to establish uniqueness. Suppose that there exist two solutions $(\lambda^{(j)},\dot\lambda^{(j)})$, $j=1,2$.
Let's define the square integrable martingale $\mathcal{M}^{(j)}$ via
\begin{align*}
&\dot\lambda^{(j)}(t) + \frac{1+\gamma+\Lambda_1}{\Lambda_2}\int_0^t\left(\frac{1+\gamma}{1+\gamma+\Lambda_1}\lambda_0^*(u)-\lambda^{(j)}(u)\right)du \\
&\quad= \frac{1+\gamma+\Lambda_1}{\Lambda_2}\E_t\left[\int_0^T\left(\frac{1+\gamma}{1+\gamma+\Lambda_1}\lambda_0^*(u)-\lambda^{(j)}(u)\right)du\right] \\
&\quad=: \mathcal{M}^{(j)}(t) \in\mathbf{L}_{\mathbb{F}}^{2}(T;\mathbb{R}),
\end{align*}
then we have that 
\begin{align*}
d \dot\lambda^{(j)} (t) = 
d\mathcal{M}^{(j)} (t)  +   \frac{(1+\gamma+\Lambda_1)}{\Lambda_2}   \left(\lambda^{(j)}(t)- \frac{1+\gamma}{1+\gamma+\Lambda_1}\lambda_0^*(t)\right) dt, 
\quad j=1,2. 
\end{align*}
Using a similar argument as in~\cite{PP90}, we can infer that $\int_0^\cdot (\dot\lambda^{(1)}(t)-\dot\lambda^{(2)}(t))d\left(\mathcal{M}^{(1)} - \mathcal{M}^{(2)}\right)(t)$ is a true martingale,  hence we have the following estimation:
\begin{align*}
&\quad\E\left[\left(\dot\lambda^{(1)}(t) - \dot\lambda^{(2)}(t)\right)^2\right]
+\E\left[\int_t^T d\langle \mathcal{M}^{(1)} - \mathcal{M}^{(2)}\rangle(u)\right] 
\\&= \frac{1+\gamma+\Lambda_1}{\Lambda_2}
\left(\E\left[\left(\lambda^{(1)}(t) - \lambda^{(2)}(t)\right)^2\right] 
- \E\left[\left(\lambda^{(1)}(T) - \lambda^{(2)}(T)\right)^2\right]\right)
\\&\leq \frac{1+\gamma+\Lambda_1}{\Lambda_2}
\E\left[\left(\lambda^{(1)}(t) - \lambda^{(2)}(t)\right)^2\right], 
\end{align*}
which further yields
\begin{align*}
\E\left[\left(\lambda^{(1)}(t) - \lambda^{(2)}(t)\right)^2\right] 
& = \E\left[\left(\int_0^t \left(\dot\lambda^{(1)}(u) - \dot\lambda^{(2)}(u)\right) \ du\right)^2\right]
\\& \leq T\E\left[\int_0^t \left(\dot\lambda^{(1)}(u) - \dot\lambda^{(2)}(u)\right)^2 du\right]
\\&\leq T\frac{1+\gamma+\Lambda_1}{\Lambda_2}
\E\left[\int_0^t \left(\lambda^{(1)}(u) - \lambda^{(2)}(u)\right)^2 du\right].
\end{align*}
From Gronwall's inequality, we can see that 
$\E\left[\left(\lambda^{(1)}(t) - \lambda^{(2)}(t)\right)^2\right]=0$ for all $t\in[0,T]$, which indicates the uniqueness. 

It is  not hard to see that the FBSDE system~\eqref{fbsde:constant} shares the same linear structure with the FBSDE system in~\cite[Lemma 3.1]{KNT20}. For a  more general linear FBSDE system, ~\cite{Y99} provides necessary and sufficient conditions for the solvability. 
\end{proof}

\subsection{Proof of Theorem~\ref{thm: optimal control logou}}~\label{app:log ou proof}
\begin{proof}
Since $\log\delta(t)$ is a Gaussian process, we have that
\begin{align}
\delta(u)
= \bar{\delta} 
\exp\left(e^{-\kappa_\delta(u-t)}\log \frac{\delta(t)}{\bar\delta} 
+ \nu_\delta e^{-\kappa_\delta(u-t)} \int_t^u e^{\kappa_\delta(v-t)} dW_\delta(v)\right)
\end{align}
Then we can easily calculate that, for $0\leq t\leq u$,
\begin{align*}
\E_t\left[\delta(u)^{-1}\right]
& = \frac{1}{\bar{\delta}} 
\exp\left(-e^{-\kappa_\delta(u-t)}\log \frac{\delta(t)}{\bar\delta} \right)
\E_t\left[ \exp\left(-\nu_\delta e^{-\kappa_\delta(u-t)} \int_t^u e^{\kappa_\delta(v-t)} dW_\delta(v)\right)\right]
\\& = \exp\left(-m_\delta(t, u) + \frac{1}{2} q_\delta (t,u)\right),
\end{align*}
and for $0\leq t\leq s\leq u$, by the moment generating function for multivariate Gaussian distribution, 
\begin{align*}
\E_t \left[\delta(u)^{-1}\sqrt{\delta(s)}\right]
& = \exp\left(\frac{1}{2}m_\delta(t,s)-m_\delta(t, u) + \frac{1}{2} q_\delta (t,u) +  \frac{1}{8} q_\delta (t,s) - \frac{1}{2} C_{\delta}(t;s,u)\right),
\end{align*}
where $m_\delta$, $q_\delta$, and $C_{\delta}$ are defined via~\eqref{eq: logou conditional mean} - \eqref{eq: logou conditional covariance}. 
Recalling that in the mean-reverting model, 
we have 
$b_\varphi(t) = \kappa_\varphi(\bar\varphi - \varphi(t))$ and 
$\sigma_\varphi^2(u)=\nu_\varphi^2\delta(u)$, then
\begin{align*}
\E_t[\lambda_0^*(u)]
=
\frac{1}{1+\gamma}
\left(
\frac{\kappa_\varphi}{\nu_\varphi^2}
\E_t\left[
\frac{\bar\varphi-\varphi(u)}{\delta(u)}
\right]
+
\frac{1+\nu_\varphi^2}{2\nu_\varphi^2}
\right).
\end{align*}
Since $\varphi(t)$ is mean-reverting, we can see that
\begin{align*}
\bar\varphi-\varphi(u)
=
(\bar\varphi-\varphi(t))e^{-\kappa_\varphi(u-t)}
-
\nu_\varphi
\int_t^u
e^{-\kappa_\varphi(u-s)}
\sqrt{\delta(s)}\,dW_\varphi(s).
\end{align*}
Therefore,
\begin{align*}
\E_t\left[
\frac{\bar\varphi-\varphi(u)}
{\delta(u)}
\right]
&= (\bar\varphi-\varphi(t))e^{-\kappa_\phi(u-t)} \E_t[\delta(u)^{-1}]
-\nu_\varphi\E_t\left[ \delta(u)^{-1} \int_t^u e^{-\kappa_\varphi(u-s)} \sqrt{\delta(s)}\,dW_\varphi(s) \right].
\end{align*}
It remains to evaluate the second conditional expectation.
Fixing \(u\), we define the martingale for \(r\in[t,u]\),
\[
\mathcal{M}(r):=\E_r[\delta(u)^{-1}]
=
\exp\left(
-m_\delta(r,u)+\frac12q_\delta(r,u)
\right).
\]
Then \(\mathcal{M}(u)=\delta(u)^{-1}\), and It\^o's formula yields
\[
d\mathcal{M}(r)
=
-\nu_\delta e^{-\kappa_\delta(u-r)}
\mathcal{M}(r)\,dW_\delta(r).
\]
Also define
\[
N(r)
=
\int_t^r
e^{-\kappa_\phi(u-s)}
\sqrt{\delta(s)}\,dW_\varphi(s).
\]
We can easily infer that $\mathcal{M}$ and $N$ are square integrable. 
Since
\(d\langle W_\phi,W_\delta\rangle(r)=\rho\,dr\),
\[
d\langle \mathcal{M},N\rangle(r)
=
-\rho\nu_\delta
e^{-(\kappa_\phi+\kappa_\delta)(u-r)}
\mathcal{M}(r)\sqrt{\delta(r)}\,dr.
\]
Applying the It\^o product formula to \(\mathcal{M}(r)N(r)\) and taking
conditional expectations gives
\[
\begin{aligned}
\E_t[\delta(u)^{-1}N(u)]
&=
-\rho\nu_\delta
\int_t^u
e^{-(\kappa_\phi+\kappa_\delta)(u-r)}
\E_t[\mathcal{M}(r)\sqrt{\delta(r)}]\,dr
\\
&=
-\rho\nu_\delta
\int_t^u
e^{-(\kappa_\phi+\kappa_\delta)(u-r)}
\E_t[\sqrt{\delta(r)}\,\delta(u)^{-1}]\,dr,
\end{aligned}
\]
where the second equality follows from the tower property.
The required integrability follows from the finite positive and
negative moments of the lognormal process \(\delta\) on finite time
intervals.

Substituting these expressions into the preceding identity yields
\eqref{eq: ce log ou}.

\end{proof}

\section{Fama--French factor attribution}
\label{sec:fama_french_test}
We examine whether the returns of the SDD-based equal-weighted tilt strategies can be explained by standard equity factors. For each portfolio $\pi$, we regress its monthly net excess return on the Fama--French factors. The three-factor specification is
\[
R^\pi(t_k)-R_f(t_k)
=
\alpha^\pi
+\beta_m^\pi\bigl(R_m(t_k)-R_f(t_k)\bigr)
+\beta_s^\pi \operatorname{SMB}(t_k)
+\beta_h^\pi \operatorname{HML}(t_k)
+\epsilon^\pi(t_k).
\]
The five-factor specification additionally includes the profitability and investment factors, RMW and CMA. To maintain consistency with the portfolio benchmark, we use the net excess return of the S\&P 500 portfolio as the market factor; the remaining factors and the risk-free rate are obtained from the Ken French Data Library. All portfolio returns are net of proportional transaction costs. The sample contains 360 monthly observations from January 1995 through December 2024.

We consider the equal-weighted portfolio and the frictional strategies generated by the mean-reverting and trending SDD specifications. Tables~\ref{table:ff3} and~\ref{table:ff5} report the three- and five-factor regression results, respectively.

\begin{table}[ht!]
\centering
\begin{tabular}{|l|c|c|c|c|c|}
\hline
 Portfolio & MKT-RF & SMB & HML & $12\widehat{\alpha}$ (\%) & $R^2$ (\%) \\
\hline
$\lambda^{*,\, \mathrm{mr}}$ & $0.984$  & $-0.003$ & $0.258$  & $2.39$  & $83.4$ \\
                      & $(41.49)$ & $(-0.10)$ & $(8.34)$ & $(1.90)$ & \\
\hline
$\lambda^{*,\, \mathrm{trend}}$  & $0.977$  & $0.054$  & $0.019$  & $1.35$  & $92.6$ \\
                      & $(65.85)$ & $(2.64)$ & $(0.97)$ & $(1.72)$ & \\
\hline
$\pi^{\mathrm{ew}} (\lambda \equiv 1)$   & $1.060$  & $0.225$  & $0.265$  & $-0.41$ & $94.4$ \\
                      & $(73.68)$ & $(11.35)$ & $(14.23)$ & $(-0.54)$ & \\
\hline
\end{tabular}
\caption{Fama--French 3-factor regression results. Monthly OLS $t$-statistics in parentheses.}
\label{table:ff3}
\end{table}
All three portfolios have market loadings close to one. In the three-factor specification, the equal-weighted portfolio has significant positive exposures to both the size and value factors. The mean-reverting strategy retains a significant value exposure but has little size exposure, while the trending strategy has a modest positive size loading and no significant value loading. Thus, the dynamic strategies materially alter the factor composition of the equal-weighted portfolio.

The annualized alpha estimates in the three-factor model are $2.39\%$ for the mean-reverting strategy and $1.35\%$ for the trending strategy, with monthly $t$-statistics of $1.90$ and $1.72$, respectively. These estimates are marginally significant at the $10\%$ level but not at the conventional $5\%$ level. The equal-weighted portfolio has a negative and statistically insignificant alpha estimate of $-0.41\%$.

Under the five-factor specification, the annualized alpha estimates decline to $1.46\%$ and $1.13\%$ for the mean-reverting and trending strategies, with $t$-statistics of $1.14$ and $1.39$. Neither estimate is statistically significant. The additional profitability and investment factors therefore explain part of the three-factor residual return, particularly for the mean-reverting strategy, while leaving the overall explanatory power of the regressions largely unchanged.
\begin{table}[ht!]
\centering
\setlength{\tabcolsep}{5pt}
\begin{tabular}{|l|c|c|c|c|c|c|c|}
\hline
Portfolio & MKT-RF & SMB & HML & RMW & CMA & $12\widehat{\alpha}$ (\%) & $R^2$ (\%) \\
\hline
$\lambda^{*,\,\text{mr}}$ & $1.010$   & $0.046$  & $0.182$  & $0.121$  & $0.093$ & $1.46$  & $83.8$ \\
                      & $(39.50)$ & $(1.23)$ & $(4.16)$ & $(2.53)$ & $(1.45)$ & $(1.14)$ & \\
\hline
$\lambda^{*, \,\text{trend}}$  & $0.987$   & $0.055$  & $-0.015$ & $0.001$  & $0.080$ & $1.13$  & $92.7$ \\
                      & $(61.36)$ & $(2.33)$ & $(-0.56)$ & $(0.04)$ & $(1.98)$ & $(1.39)$ & \\
\hline
$\pi^{\mathrm{ew}} (\lambda\equiv 1)$   & $1.077$   & $0.250$  & $0.212$  & $0.058$  & $0.084$ & $-0.96$ & $94.5$ \\
                      & $(69.48)$ & $(10.90)$ & $(7.99)$ & $(1.99)$ & $(2.17)$ & $(-1.23)$ & \\
\hline
\end{tabular}
\caption{Fama--French 5-factor regression results. Monthly OLS $t$-statistics given in parentheses.}
\label{table:ff5}
\end{table}

Overall, the factor regressions show that the SDD strategies generate exposures that differ meaningfully from those of the equal-weighted portfolio. The positive alpha estimates suggest that their out-of-sample returns are not fully captured by the three-factor model, but the evidence is not robust across factor specifications. We therefore interpret these regressions as descriptive factor attribution rather than conclusive evidence of persistent abnormal returns.

\section*{Declarations}
The authors used AI coding agents (Claude Code and Codex) to assist with the coding of experiments conducted in this paper, as well as the generation of figures and tables. Large language models (Claude and ChatGPT) were used to improve the clarity of writing in certain paragraphs. All content has been reviewed and revised by the authors. The authors take full responsibility for the accuracy and originality of the work. 

\clearpage
\bibliographystyle{plain}
\bibliography{main}
\end{document}